\documentclass[smallextended]{svjour3}       % onecolumn (second format)
\smartqed  % flush right qed marks, e.g. at end of proof

\usepackage{amsmath}
\usepackage{tabularx}
\usepackage{graphicx}
\usepackage{subfigure}
\usepackage{amssymb}
\usepackage{mathptmx}
\usepackage{epstopdf}
\usepackage{color}
\usepackage{cite}
\usepackage{algorithm,algpseudocode}
\newtheorem{assumption}{Assumption}
\begin{document}
\title{Dynamic Contract Design for Systemic Cyber Risk Management of Interdependent Enterprise Networks%\thanks{Grants or other notes
}

\titlerunning{Dynamic Contract Design for Systemic Cyber Risk Management}        % if too long for running head

\author{Juntao Chen         \and
        Quanyan Zhu \and Tamer Ba\c{s}ar %etc.
}


\institute{J. Chen and Q. Zhu \at
              Department of Electrical and Computer Engineering, Tandon School of Engineering\\
               New York University, Brooklyn NY 11201 USA \\
              \email{\{jc6412,qz494\}@nyu.edu}           %  \\
           \and
           T. Ba\c{s}ar \at
            Coordinated Science Laboratory, University of Illinois at Urbana-Champaign\\
            Urbana, IL 61801 USA\\
            \email{basar1@illinois.edu}          
}

\date{Received: date / Accepted: date}

\maketitle

\begin{abstract}
The interconnectivity of cyber and physical systems and Internet of things has created ubiquitous concerns of cyber threats for enterprise system managers. It is common that the asset owners and enterprise network operators need to work with cybersecurity professionals to manage the risk by remunerating them for their efforts that are not directly observable. In this paper, we use a principal-agent framework to capture the service relationships between the two parties, i.e., the asset owner (principal) and the cyber risk manager (agent). Specifically, we consider a dynamic systemic risk management problem with asymmetric information where the principal can only observe cyber risk outcomes of the enterprise network rather than directly the efforts that the manager expends on protecting the resources. Under this information pattern, the principal aims to minimize the systemic cyber risks by designing a dynamic contract that specifies the compensation flows and the anticipated efforts of the manager by taking into account his incentives and rational behaviors. We formulate a bi-level mechanism design problem for dynamic contract design within the framework of a class of stochastic  differential games. We show that the principal has rational controllability of the systemic risk by designing an incentive compatible estimator of the agent's hidden efforts. We characterize the optimal solution by reformulating the problem as a stochastic optimal control program which can be solved using dynamic programming. We further investigate a benchmark scenario with complete information and identify conditions that yield zero information rent and lead to a new certainty equivalence principle for   principal-agent problems. Finally, case studies over networked systems are carried out to illustrate the theoretical results obtained.
\keywords{Systemic Risk \and Dynamic Contracts \and Differential Games \and Internet of Things \and Economics of Cybersecurity}
\end{abstract}

\section{Introduction}
Cybersecurity is a critical issue in modern enterprise networks due to the adoption of advanced technologies, e.g., Internet of things (IoT), cloud and data centers, and supervisory control and data acquisition (SCADA) system, which create abundant surfaces for cyber attacks  \cite{knowles2015survey,chen2017security,sicari2015security}. Due to the interconnections between nodes in the network, the cyber risk can propagate and escalate into systemic risks, which have been a major contributor to massive spreading of Mirai botnets, phishing messages, and ransomware, causing information breaches and financial losses.  In addition, systemic risks are highly dynamic by nature as the network faces a continuous flow of cybersecurity incidents. Hence, it becomes critical for the network and asset owner to protect resources from cyber attacks.

The complex interdependencies between nodes and fast evolution nature of threats have made it challenging to mitigate systemic risks of enterprise network and thus requires expert knowledge from cyber domains. The asset owners or system operators need to delegate tasks of risk management including security hardening and risk mitigation to security professionals. As depicted in Fig. \ref{security_delegation_pic}, the owner can be viewed as a principal who employs a security professional to fulfill tasks that include monitoring the network, patching the software and devices, and recovering machines from failures. The security professionals can be viewed as an agent whose efforts are remunerated by the principal. This principal-agent type of interaction models the service relationships between the two parties. The effort of the agent can be measured by the hours he spends on the security tasks. Moreover, the amount of allocated effort has a direct impact on the systemic cyber risk. For example, with more frequent scans on suspicious files and the Internet traffic at each node, the cyber risk becomes low and less likely to spread. An agent plays an important role in systemic risk as he can determine the amount of his effort and the way of distributing efforts on protecting nodes over the network. Hence, it is essential for the principal to incentivize the agent to distribute his resources desirably to protect the network. 

\begin{figure}[t]
  \centering
    \includegraphics[width=0.75\columnwidth]{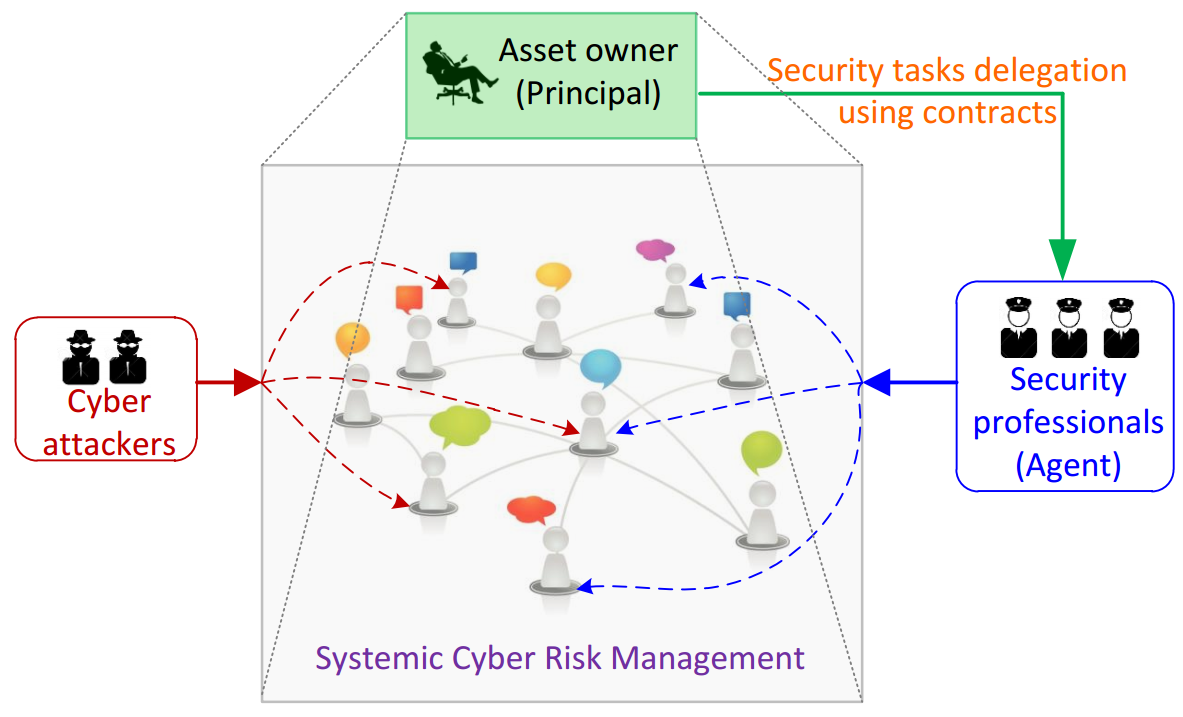}
  \caption[]{Systemic cyber risk management for enterprise network. The asset owner (principal) delegates the risk management tasks, e.g., network monitoring and software patching, to security professionals (agents) by designing a contract which specifies the remuneration schemes. The amount of remuneration is directly related to the systemic risk outcome of the network.}
  \label{security_delegation_pic}
\end{figure}

In the cyber risk management of enterprise network, one distinction is the lack of knowledge of the principal about the effort spent by the agent. The principal is only able to observe risk outcomes, e.g., the denial or failures of services and conspicuous performance degradation. Moreover, due to the randomness in the cyber network, e.g., the biased assessment of risks and the unknown attack behaviors, the cyber risk evolves under uncertainties, making it difficult for the principal to infer the exact effort of the agent from the observations.  This type of incomplete information structure is called moral hazard in contracts, under which the asset owner aims to minimize the systemic cyber risk by providing sufficient incentives to the risk manager through a dynamic contract that specifies the compensation flows and suggested effort, while the risk manager's objective is to maximize his payoff with minimum effort by responding to the agreed contract.

The dynamic principal-agent problem has an asymmetric information structure in which the risk manager determines his effort over time, while  this effort is hidden to or unobservable by the asset owner. This information structure makes the contract design a challenging decision making problem.  Conventional methods to address problems of incomplete information include information state based separation principle \cite{james1994risk,james1996partially} and belief update scheme \cite{cho1987signaling}. However, these methods cannot be directly applied to design an optimal contract for the players. To address this challenge, we develop a systematic solution methodology which includes an estimation phase, a verification phase, and a control phase. Specifically, we first anticipate the risk manager's optimal effort based on the systemic risk outcome by  designing an estimator for the principal. Then, we show that the principal has \textit{rational controllability} of the systemic risk by verifying that the estimated effort is incentive compatible. Finally, we transform the problem using decision variables that adapt to the principal's information set and obtain the solution by solving a reformulated standard stochastic control program.

The optimal dynamic mechanism design (ODMD) includes the compensation flows and the suggested effort. The designed optimal dynamic contract includes the compensations for direct cost of effort, discounted future revenue, cyber risk uncertainties, as well as incentive provisions. Furthermore, under the incentive compatible contract, the risk manager's behavior is strategically neutral in the sense that his current action depends solely on the present stage's cost. The policies of the optimal contract can be determined by solving a stochastic optimal control problem. Under mild conditions, the decision variables associated with the suggested effort and the compensation can be solved in parallel, leading to a \textit{separation principle} for dynamic mechanism design. 
As a benchmark problem for comparison, we further investigate the dynamic contract under full information where the principal can fully observe the agent's effort. In general cases, we show that there is a positive information rent quantifying the difference of principal's objective value between the contracts designed under incomplete information and full information. 
In addition, we identify conditions under which the information rent is degenerated to zero, yielding a \textit{certainty equivalence principle} in which the mechanism designs under full and asymmetric information become identical. 
For example, the hidden-action impact is absent in the linear quadratic (LQ) case where the principal achieves a perfect estimation and control of the risk manager's dynamic effort.

The incentive provided by the principal to the agent is critical for mitigating the cyber risk. Without sufficient control effort, the risk would grow and propagate over the network. Under the optimal dynamic contract, both the systemic cyber risk and adopted effort decrease over time. Moreover, the effort converges to a positive constant and the systemic risk can remain at a low level. Furthermore, a higher network connectivity requires the agent to spend more effort to reduce the systemic cyber risk. In the linear quadratic (LQ) scenario, we observe that the nodes in the cyber network have self-accountability, i.e., the amount of effort allocated on each node depends only on its risk influences on other nodes and is independent of exogenous risks coming from neighboring nodes. This observation enables large-scale implementation of distributed risk mitigation policy by determining the outer degrees of the nodes.

The contributions of this work are summarized as follows. 
\begin{itemize}
\item[1)] We formulate a dynamic mechanism design problem for systemic cyber risk management of enterprise networks under hidden-action type of incomplete information.
\item[2)] We provide a systematic methodology to characterize the optimal mechanism design by transforming the problem into a stochastic optimal control problem with compatible information structures.
\item[3)] We define the concept of ``rational controllability'' to capture the feature of indirect control of cyber risks by the principal, and identify the explicit conditions under which the designed dynamic contract is incentive compatible.
\item[4)] We identify a separation principle for dynamic contract design under mild conditions, where the estimation variable capturing the suggested risk management effort and the control variable specifying the compensation can be determined separately.
\item[5)] We reveal a certainty equivalence principle for a class of dynamic mechanism design problems where the information rent is zero, i.e., the contracts designed under asymmetric and full information cases coincide. 
\item[6)] We observe that larger enterprise network connectivity and risk dependency strength require the principal to provide more incentives to the agent. Under the optimal contract in the LQ case, the allocated effort depends on the nodes' outer degree, leading to a self-accountable and distributed risk mitigation scheme.
\end{itemize}

\subsection{Related Work}
Cybersecurity becomes a critical issue due to the large-scale deployment of smart devices and their integration with information and communication techologies (ICTs) \cite{sicari2015security,pawlick2019istrict}. Hence, security risk management is an important task which has been investigated in different research fields, such as communications and infrastructures \cite{zhu2012interference,chen2019dynamicgame}, cloud computing \cite{takabi2010security} and IoT \cite{chen2019interdependent}. The interconnections between nodes and devices make the risk management a challenge problem as the cyber risk can propogate and escalate into systemic risk \cite{fouque2013handbook}, and hence the interdependent security risk analysis is necessary \cite{chen2019optimalsecure}. Managing systemic risk is nontrivial as demonstrated in financial systems \cite{bisias2012survey}, critical infrastructures \cite{crowther2005application}, and communication networks \cite{cherdantseva2016review}.  In a network with a small number of agents, graph-theoretic methods have been widely adopted to model the strategic interactions and risk interdependencies between agents  \cite{bisias2012survey,elliott2014financial}. When the number of nodes becomes large, \cite{carmona2015mean} has proposed a mean-field game approach where a representative agent captures the system dynamics. Different from \cite{eisenberg2001systemic,acemoglu2015systemic} in minimizing the static systemic risk at equilibrium, we focus in this paper on a mechanism design problem that can reduce the systemic risks by understanding the system dynamics.

Dynamic games of incomplete or imperfect information have been studied within the context of different classes of games, such as repeated games \cite{aumann1995repeated}, differential games \cite{cardaliaguet2007differential}, and stochastic games \cite{zhu2010heterogeneous}. Many types of information structures that entail incomplete or imperfect information have been investigated in the literature, such as partial or noisy measurements of system states \cite{james1994risk,james1996partially,
hansen2004dynamic,
basar1985equilibrium,bacsar2014stochastic}, and asymmetric information for the players \cite{cardaliaguet2009continuous,gupta2014common,
gupta2016dynamic}. Approaches to control and optimization under classical information structures,  also extended to games, include the information state based separation principle \cite{james1994risk,james1996partially,
charalambous1997role}, belief updates on players' private information \cite{cho1987signaling}, generalized belief states of agents \cite{hansen2004dynamic}, and control over networks \cite{yuksel2013stochastic}.
Decision-making under nonclassical information structures has also been studied (such as \cite{srikant1992asymptotic,bansal1987stochastic,
bacsar1994optimum}), where the players are coupled through the system dynamics and/or the performance indices do not share the same information and could be memoryless. Our bi-level dynamic mechanism design problem exhibits a unique information structure in that the principal delegates the risk control tasks to the agent without observing the applied control effort, while the agent has complete information of the system, which leads to informational asymmetry.

Dynamic mechanism design has been studied broadly \cite{gershkov2014dynamic,athey2013efficient}. In \cite{cvitanic2013contract}, the authors have provided a comprehensive summary of dynamic contract design based on the stochastic maximum principles where solving forward-backward stochastic differential equations (FBSDEs) becomes necessary. Instead of controlling the output density using Girsanov transform, which has an indirect interpretation in applications \cite{schattler1993first,
 williams2015solvable}, the authors in \cite{sannikov2008continuous,biais2010large} directly control the system output and adopt an alternative approach by regarding the agent's future payoff as a variable in the stochastic control dynamics. Our approach in this paper adopts the agent's current payoff as a state variable which is different from the above discussed methods. The purpose of this work is to develop risk management solutions for networked systems using a remuneration scheme that combines intermediate and terminal compensations. We will develop a systematic solution methodology for this class of problems by capturing the systemic cyber risk dynamics and provide principles for optimal mechanism design.
 
The current work is different from the preliminary version \cite{chen2018linear} in multiple aspects. First and foremost, in \cite{chen2018linear}, the risk management policy is designed only for the LQ framework, while the current one extends the model to arbitrarily general scenarios. Thus, the analysis and derived results in this work are much more fundamental by focusing on a broader class of dynamic contract design problems. Second, we additionally investigate the dynamic contract design under full information for comparison and obtain a new certainty equivalence principle for a number of scenarios. Third, we provide comprehensive motivations for the established dynamic risk model in the problem formulation, and include discussion and illustration on the timing of events during contract design. Fourth, the introduction section is substantially expanded, including depiction of risk management for enterprise networks, background on systemic risks, and description of explicit contributions of the work. Fifth, we enrich the related work section completely by discussing more literature and highlight the differences. Sixth, we include a higher number of case studies to thoroughly illustrate the dynamic contract design principles for systemic cyber risk management in enterprise networks.

\subsection{Organization of the Paper}
The paper is organized as follows. We formulate the systemic cyber risk management problem in Section \ref{formulation}. Section \ref{analysis} analyzes the dynamic contract forms and the incentive constraints. Section \ref{principal_problem} reformulates the principal's problem and solves a linear quadratic case  explicitly. Section \ref{benchmark} presents
a complete-information benchmark scenario for comparison. Section \ref{examples} presents examples to illustrate the dynamic contract design for systemic risk management. Section \ref{conclusion} concludes the paper.

\section{Problem Formulation}\label{formulation}
This section formulates the dynamic systemic cyber risk management problem of enterprise networks under asymmetric information using a principal-agent framework, and presents an overview of the adopted methodology.

\subsection{Systemic Cyber Risk Management}
An enterprise network is comprised of a set $\mathcal{N}$ of  nodes, where $\mathcal{N}=\{1,2,...,N\}$. Due to the interdependencies among different nodes and fast changing nature of the threats, mitigating the systemic cyber risk is a challenging task which requires expertise from cybersecurity professionals. For example, to reduce the enterprise network vulnerability, it requires a constant monitoring of the Internet traffic into and out of the system, regular patching and updating of the device software, and continuous traffic scanning for intrusion detection. The principal\footnote{The principal refers to the network/asset owner, and the agent refers to the risk manager or security professional which are used interchangeably.} can delegate the risk management tasks over a time period $[0,T]$ to a professional manager.
 
 \begin{figure}[t]
  \centering
    \includegraphics[width=.65\columnwidth]{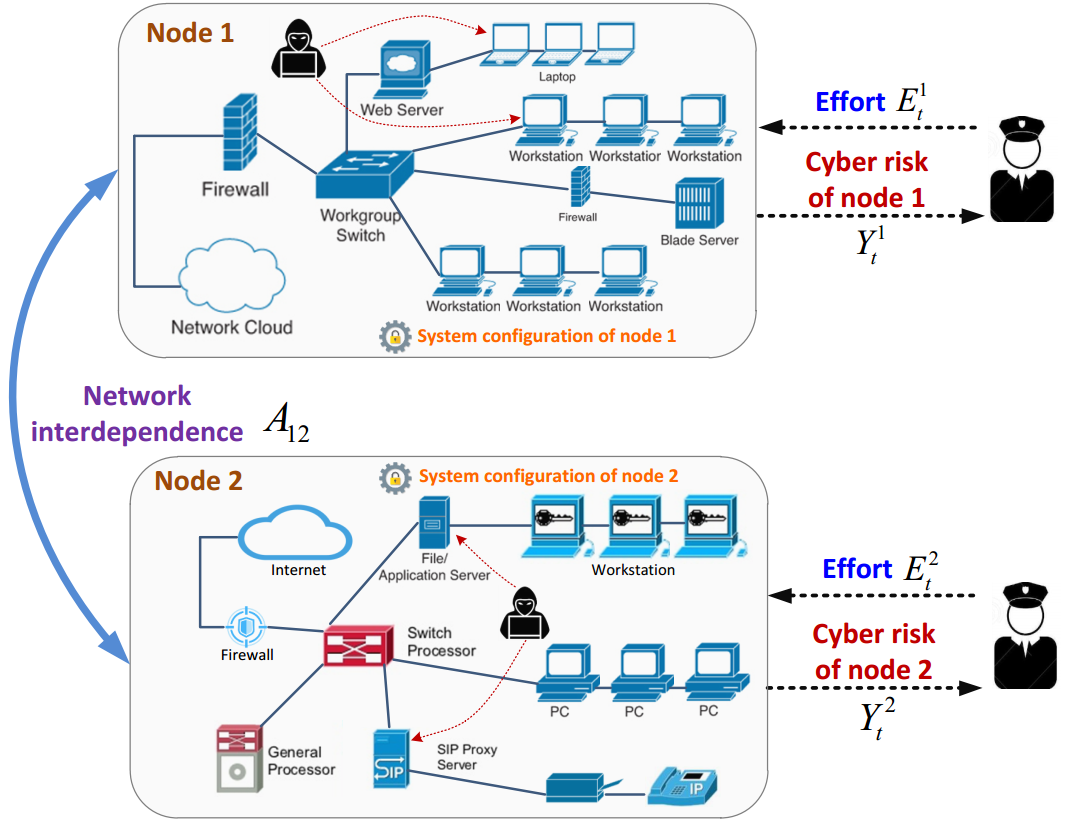}
  \caption[]{Systemic cyber risk management of an enterprise network containing two nodes. The cyber risk at node $i$ is denoted by $Y_t^i$ and the applied risk manager's effort is $E_t^i$, $i\in\{1,2\}$. The cyber risk at each node depends on its system configuration, the attack model, and the risk manager's effort. Note that the cyber risk can propagate due to the connections between nodes.}
  \label{cyber_effort_illustration}
\end{figure}

The cyber risk of each node depends on the level of compliance with security criteria, the number of vulnerabilities of the software and hardware assets, the system configurations, and the concerned threat models \cite{refsdal2015cyber}. The risk also evolves over time as the enterprise node constantly updates its software, introduces new functionalities, and interconnects with other nodes.  We let $Y^i_t \in \mathbb{R}$ be the state of node $i\in\mathcal{N}$ to capture the risk of each node that maps the system configurations at time $t$ and the threat models to the associated risk. For example, under the advanced persistent threat (APT) type of cyber attacks, one can assess the node's risk using FlipIt game model in which the defender strategically configures the system by reclaiming the control of the node with some frequencies \cite{van2013flipit}. The FlipIt game outcome yields node's risk which is the expected proportion of time that the node may be compromised by the adversary. As the nodes in the enterprise network are connected, their risks become interdependent. We use an $N\times N$-dimensional real matrix $A$ with non-negative entries to model the influence of node $i$ on node $j$, $i, j \in \mathcal{N}$. The diagonal entries in $A$ represent the strength of internal risk evolution, and the off-diagonal entries capture the risk influence magnitude between nodes \cite{nguyen2009stochastic,chen2019interdependent}. For convenience, the risk profile of the network is denoted by $Y_t = [Y_t^1, Y_t^2, \cdots, Y_t^N]$. The dynamics of the risk profile describes the evolution of the systemic risk of the whole network.

 To manage the risk profile, the risk manager can apply effort continuously over the time  period $[0,T]$. Specifically, at every time $t$, $t\in[0,T]$, the risk manager can spend effort $E_t\in\mathcal{E}\subseteq\mathbb{R}_+^N$ on the nodes that mitigates the systemic cyber risk, where $\mathcal{E}$ is a compact set. As fore-mentioned, the effort can be measured by the amount of time and effectiveness of the risk manager spent on monitoring the cyberspace of the enterprise network.  The amount of reduced risk is monotonically increasing with the allocated effort $E_t$ \cite{miura2008security}. This fact is reflected by many security practices, e.g., frequent scanning and analyzing the log files as well as timely patching the software can reduce the probability of successful cyber compromise by the adversary. Another critical factor to be considered is that the cyber risk faces uncertainties due to the randomness in the cyber network, e.g., the biased assessment and measurement of risk losses and under-modelling of random cyber threats \cite{li2010uncertainty}. Similar to \cite{carmona2015mean}, we use an $N$-dimensional standard Brownian motion $B_t$ which is defined on the complete probability space $(\Omega,\mathcal{F},\mathbb{P})$ to model the risk uncertainties on nodes. For clarity, Fig. \ref{cyber_effort_illustration} depicts an example of cyber risk management of the enterprise network containing two interdependent nodes. Each node stands for a subnetwork with its own system configuration, and the adversary can target different assets, e.g., application servers and workstations. The risk manager applies efforts $E_t^1$ and $E_t^2$ to node 1 and node 2 continuously to reduce the cyber risks $Y_t^1$ and $Y_t^2$, respectively. The interdependency between two nodes is captured by the factor $A_{12}=A_{21}$.
 
  In sum, we focus on a model of systemic cyber risk evolution described by the following stochastic differential equation (SDE):
\begin{equation}\label{cyber_risk}
\begin{split}
dY_t &= AY_tdt - E_t dt + \Sigma_t(Y_t) dB_t,\\
Y_0 &= y_0,
\end{split}
\end{equation}
where $y_0\in\mathbb{R}_{+}^N$ is a known positive vector denoting the initial systemic risk. Let $\mathbb{D}_+^{N\times N}$ denote the space of diagonal real matrices with positive elements. Then, $\Sigma_t:\mathbb{R}^N\rightarrow\mathbb{D}_+^{N\times N}$ captures the volatility of cyber risks in the network. Here, the diffusion coefficient $\Sigma_t(Y_t)$ indicates that the magnitude of uncertainty can be related to the dynamic risk of each node. We assume that the entries in $\Sigma_t(Y_t)$ are bounded, satisfying $\int_0^T \Vert\Sigma_t(Y_t)\mathbf{1}_N\Vert^2 dt \leq C_1$ almost surely, where $C_1$ is a positive constant, $\Vert \cdot \Vert$ denotes the standard Euclidean norm, and $\mathbf{1}_N$ is an $N$-dimensional vector with all ones. Furthermore, the risk manager's effort $E_t$ satisfies the condition $\int_0^T \vert E_t\vert dt \leq C_2$ almost surely, where $C_2$ is a positive constant. Since the manager can apply effort to every node through $E_t$, the systemic risk level $Y_t$ is fully manageable in the sense that more effort on each node reduces its cyber risk more significantly. Note that the model in \eqref{cyber_risk} captures the characteristics of systemic cyber risks of enterprise network, and it is also adopted in various others' risk management scenarios inluding cyber-physical industrial control systems \cite{zhu2012dynamic} and financial networks \cite{garnier2013diversification}.

\begin{figure}[t]
  \centering
  
{%
    \includegraphics[width=1\textwidth]{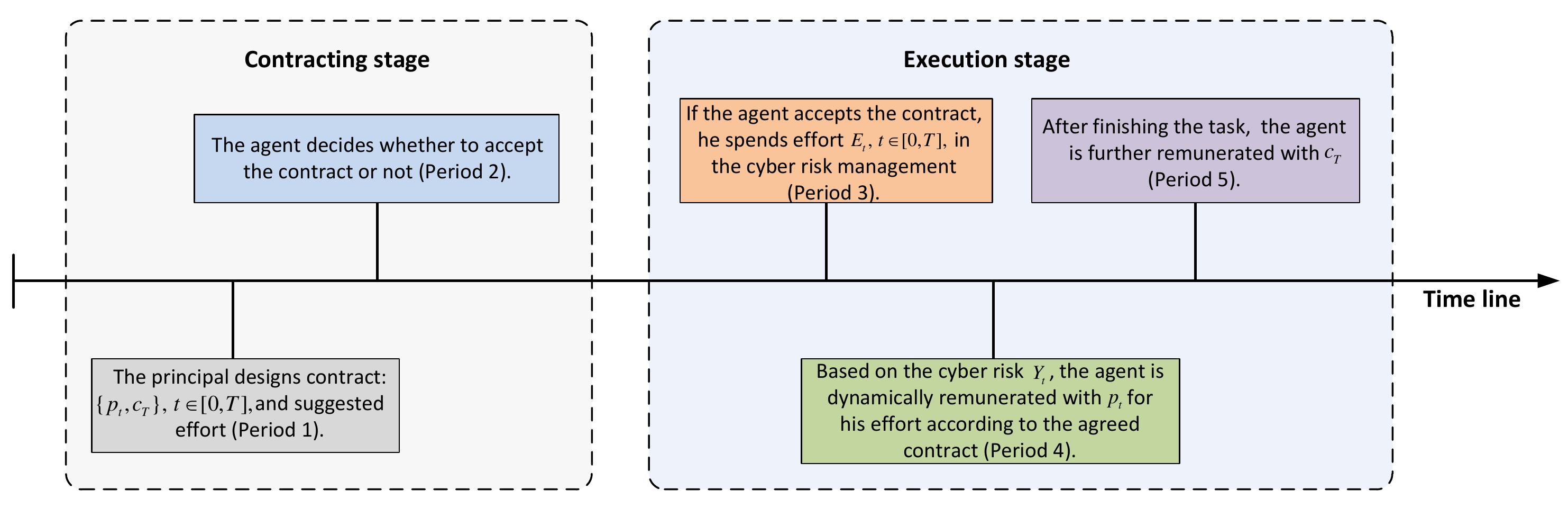}%
    \label{fra}%
  }%
  \caption{Timeline of the dynamic contract design for systemic cyber risk management.}
  \label{timeline_contract}
\end{figure}

As shown in Fig. \ref{timeline_contract}, the dynamic contract design for cyber risk management can be broken into two stages, namely the contracting stage and the execution stage. In the contracting stage, the principal first provides a dynamic contract that specifies the payment rules for the risk management to the agent and suggested/anticipated effort. Then, the agent chooses to accept the contract or not based on the provided benefits. If the agent accepts, then at the execution stage he needs to determine the adopted effort $E_t$ to reduce the systemic cyber risk. During the task, the principal observes the dynamic risk outcome $Y_t$ and pays $p_t\in\mathcal{P}\subseteq \mathbb{R}_+$ compensation to the agent according to the agreed contract, where $\mathcal{P}$ is a compact set.
 After completing the task, the agent also receives a terminal payment $c_T\in\mathbb{R}_{+}$ which finalizes the contract. 
 
Therefore, the principal needs to decide on the payment process $\{p_t\}_{0\leq t\leq T}$ as well as the final compensation $c_T$ by observing the systemic risks. Note that the effort level $E_t$, $t\in [0,T]$, is hidden information of the agent, which corresponds to the hidden-action scenario,  or moral hazard, in contract theory. This feature a reflection of the fact that the principal (asset owner) of the  enterprise network cares about the cyber risk outcome $Y_t$ rather than the implicit effort $E_t$ adopted by the risk manager. Furthermore, we denote the principal's information set by $\mathcal{Y}_t$, representing the augmented filtration generated by $\{Y_s\}_{0\leq s\leq t}$. The agent's information set is denoted by $\mathcal{A}_t$, including $\{Y_s\}_{0\leq s\leq t}$ and $\{B_s\}_{0\leq s\leq t}$. Note that for the agent, knowing $\{Y_s\}_{0\leq s\leq t}$ or $\{B_s\}_{0\leq s\leq t}$ is equivalent as he can determine one based on the other using also his effort process $\{E_s\}_{0\leq s\leq t}$.  Specifically, at time $t$, the principal's knowledge includes only the path of $Y_s$, $0\leq s\leq t$. In comparison, the agent can observe every term in the system, including the principal's information as well as the path of $B_s$, $0\leq s\leq t$. The principal observes risk outcome $Y_t$, and his goal is to reduce the systemic risk by providing incentives to the manager. Therefore, the principal has no direct control of the systemic risk, and the difficulty he faces is in designing an efficient remuneration scheme based only on the limited observable information.

Next, we rewrite the $\mathcal{Y}_T$-measurable terminal payment as $c_T = \int_0^T dc_t+c_0$, to facilitate the contract analysis,  where $c_t$ has an interpretation of cumulative payment during $[0,t]$, and $c_0$ is a constant to be determined. Note that $c_0$ is a virtual initial payment and the agent receives it not at initial time 0, but rather at the  terminal time $T$ which is captured by the term $c_T$.
The evolution of the aggregated equivalent $\mathcal{Y}_t$-measurable financial income  process $M_t$ of the cyber risk manager can be described by
\begin{equation}\label{M_t}
dM_t = dc_t + p_tdt.
\end{equation}

The cyber risk manager's cost function is:
\begin{equation}\label{J_A}
\begin{split}
J_A\left(\{E_t\}_{0\leq t\leq T};\{p_t\}_{0\leq t\leq T},c_T\right) = \mathbb{E}\int_0^T e^{-rt} f_A(t,p_t,E_t) dt + e^{-rT}h_A(M_T),
\end{split}
\end{equation}
where $\mathbb{E}$ is the expectation operator, $r\in\mathbb{R}_+$ is a discount factor, $f_A:[0,T]\times  \mathbb{R}_+\times \mathcal{E}\rightarrow\mathbb{R}$ is the running cost, and $h_A: \mathbb{R}_+\rightarrow\mathbb{R}_-$ is the terminal cost. The function $f_A$ is (implicitly) composed of two terms: the cost of spending effort $E_t$ in risk management, and the received compensation $p_t$ from the principal. Note that the final compensation $c_T$ is incorporated into $h_A(M_T)$. Assumptions we make on the two additive terms of the cost functions are as follows.

\begin{assumption}\label{Assump_1}
The running cost function $f_A(t,p_t,E_t)$ is uniformly continuous and differentiable in $p_t$ and $E_t$. Further, it is monotonically decreasing in $p_t$, and monotonically increasing and strictly convex in $E_t$. The terminal cost function $h_A(M_T)$ is a continuously differentiable, convex, and monotonic decreasing function.
\end{assumption}

The principal's cost function, on the other hand, is specified as:
\begin{equation}\label{J_P}
\begin{split}
J_P(\{p_t\}_{0\leq t\leq T},c_T) 
= \mathbb{E}\int_0^T e^{-rt} f_P(t,Y_t,p_t)dt + e^{-rT}\left( c_T+ h_P(Y_T)\right),
\end{split}
\end{equation}
where $f_P:[0,T]\times \mathbb{R}^N\times \mathcal{P}\rightarrow\mathbb{R}$ is the running cost, and $h_P: \mathbb{R}^N\rightarrow\mathbb{R}$ denotes the terminal cost. The function $f_P$ captures the instantaneous cost of dynamic systemic risk and the payment to the agent.

\begin{assumption}\label{Assump_2}
The running cost for the principal, $f_P(t,Y_t,p_t)$,  is uniformly continuous and  differentiable in $Y_t$ and $p_t$. Further, it is monotonically increasing in $p_t$ and $Y_t$. The terminal cost for the principal, $h_P(Y_T)$, is a  continuously differentiable and  monotonic increasing function.
\end{assumption}

\subsection{Dynamic Principal-Agent Model}
In cyber risk management, the principal contracts with the agent over $[0,T]$. For a given contract, the risk manager is strategic in minimizing the net cost. This rational behavior can be captured by the following definition.
 
\begin{definition}[Incentive Compatibility]
Under a given payment process $\{p_t\}_{0\leq t\leq T}$ and terminal compensation $c_T$ of the principal, the effort trajectory $\{E_t^*\}_{0\leq t\leq T}$ of the agent is incentive compatible (IC) if it optimizes the  cost function \eqref{J_A}, i.e.,
\begin{equation}\label{IC_eqn}
\begin{split}
& J_A\left(\{E_t^*\}_{0\leq t\leq T};\{p_t\}_{0\leq t\leq T},c_T\right) \\
& \leq  J_A\left(\{E_t\}_{0\leq t\leq T};\{p_t\}_{0\leq t\leq T},c_T\right),
  \forall E_t\in\mathcal{E},\ t\in[0,T].
 \end{split}
 \end{equation}
\end{definition}
 
The asset owner needs to provide sufficient incentives for the agent to fulfill the task of risk management, and this fact is captured through individual rationality as follows.

\begin{definition}[Individual Rationality]
The agent's policy is individually rational (IR) if the effort trajectory $\{E_t^*\}_{0\leq t\leq T}$ leads to satisfaction of 
\begin{equation}\label{IR_eqn}
\begin{split}
J_A\left(\{E_t^*\}_{0\leq t\leq T};\{p_t\}_{0\leq t\leq T},c_T\right) = \inf_{E_t\in\mathcal{E}} J_A\left(\{E_t\}_{0\leq t\leq T};\{p_t\}_{0\leq t\leq T},c_T\right) \leq \underline J_A, 
\end{split}
\end{equation}
where $\underline{J}_A$ is a predetermined non-positive constant.
\end{definition}
Note that the non-positiveness of $\underline J_A$ ensures the profitability of risk manager by fulfilling the risk management tasks.

We next provide precise formulations of the problems faced by the agent and the principal. Under a contract  $\{\{p_t\}_{0\leq t\leq T},c_T\}$, the agent minimizes his total cost by solving the following problem:
\begin{align*}
\mathrm{(O-A)}:\quad &\min_{E_t\in\mathcal{E},\ t\in[0,T]}\ J_A\left(\{E_t\}_{0\leq t\leq T};\{p_t\}_{0\leq t\leq T},c_T\right)\\
\mathrm{subject\ to}\  &\mathrm{the\ stochastic\ dynamics}\ \eqref{cyber_risk},\ \mathrm{and\ the\ payment\ process}\ \eqref{M_t}.
\end{align*}

By taking into account the IC and IR constraints, the principal addresses the following optimization problem:
\begin{align*}
\mathrm{(O-P)}:\ \ &\min_{p_t\in\mathcal{P},\ t\in[0,T],\ c_T}\ J_P(\{p_t\}_{0\leq t\leq T},c_T)\\
\mathrm{subject\ to}\  &\mathrm{the\ stochastic\ dynamics}\ \eqref{cyber_risk},\ \mathrm{IC}\ \eqref{IC_eqn},\ \mathrm{and\ IR}\ \eqref{IR_eqn}.
\end{align*}  
Note that the designed contract terms $\{p_t\}_{0\leq t\leq T}$ and $c_T$ should adapt to the information available to the principal in view of the underlying  incomplete information.
Denote the solution to $\mathrm{(O-P)}$ by $\{p_t^*\}_{0\leq t\leq T}$ and $c_T^*$. We present the solution concept of the formulated problem as follows.

\begin{definition}[Optimal Dynamic Mechanism Design (ODMD)]
The ODMD consists of the contract $\{\{p^*_t\}_{0\leq t\leq T},c_T^*\}$ as well as the effort process $\{E_t^*\}_{0\leq t\leq T}$ that solve the problems $\mathrm{(O-P)}$ and $\mathrm{(O-A)}$, respectively. In addition, the compensation processes $p^*_t$ and $c_T^*$ are adapted to $\mathcal{Y}_t$ and $\mathcal{Y}_T$, respectively, and the risk manager's effort $E_t^*$ is adapted to $\mathcal{A}_t$.
\end{definition}

\textit{Remark:} ODMD captures the bi-level interdependent decision making of the principal and the agent, which is a Stackelberg differential game with a nonstandard information structure. Since the principal (leader) delegates the control task to the agent (follower) but cannot observe his adopted action, ODMD features the limited nature of the principal's information.

Due to the hidden effort of the risk manager, $\mathrm{(O-P)}$ is not a classical stochastic optimal control problem. Specifically, the principal only observes the cyber risk outcome rather than the effort which has to be incentivized. To address this challenge brought about by the presence of asymmetric information, we adopt a systematic approach to design an incentive compatible and optimal mechanism.

\subsection{Overview of the Methodology}
We present an overview of the steps involved in our derivation, with details worked out in the following sections. 
 
The principal first estimates the risk manager's effort based on the systemic risk output (estimation phase), and then verifies that the estimated effort is incentive compatible (verification phase), and finally designs an optimal compensation scheme under the incentive compatible estimator (control phase). To address the challenge, our goal is to transform the problem using variables that adapt to the principal's information set. To this end, the principal first assumes that the agent behaves optimally with effort level $E_t^*$ (even though the principal does not know the exact value) and calculates the corresponding cost of the agent. Another interpretation for this step would be that the principal anticipates the agent implementing $E_t^*$ which satisfies the IC constraint. Then, the principal designs the terminal payment form using the estimated agent's cost (Section \ref{terminal_analysis}).  The agent responds to the contract strategically through his best effort $E_t^o$. When the anticipated $E_t^*$ coincides with $E_t^o$, $E_t^*$ is an incentive compatible estimator and the principal facilitates the agent implementing $E_t^*$ successfully (Section  \ref{incentive_analysis}). Therefore, the principal can determine the optimal payment $p_t^*$ based on $E_t^*$ by solving a standard stochastic optimal control problem (Section \ref{principal_problem_reform}).

\section{Analysis of Risk Manager's Incentives}\label{analysis}
We first provide a form of the terminal payment contract term and then focus on deriving an incentive compatible estimator of the cyber risk manager's effort.

\subsection{Terminal Payment Analysis}\label{terminal_analysis}
We first present the following result on the IR constraint. 
\begin{lemma}
The IR constraint holds as an equality, i.e., 
\begin{equation}
J_A\left(\{E_t^*\}_{0\leq t\leq T};\{p_t\}_{0\leq t\leq T},c_T\right) = \underline J_A.
\end{equation}
\end{lemma}
\begin{proof}
If $J_A\left(\{E_t^*\}_{0\leq t\leq T};\{p_t\}_{0\leq t\leq T},c_T\right) < \underline J_A$, the designed contract is not optimal as the principal can further reduce his cost by paying less to the agent. \qed
\end{proof}

Next, we first express the agent's cost under the principal's information set $\mathcal{Y}_t$ as well as using the property that the agent chooses an optimal $E_t^*$, and then use the principal's estimation about the agent's cost to characterize the cumulative payment process.
We introduce a new variable $W_t$ representing the expected future cost of the agent anticipated by the principal as follows:
\begin{equation}\label{Wt}
W_t = \mathbb{E}\left[\int_t^T e^{-r(s-t)} f_A\big(s,p_s,E_s^*) ds + e^{-r(T-t)}h_A(M_T)\big| \mathcal{Y}_t\right].
\end{equation}
Note that $W_t$ is evaluated under the information available to the principal at time $t$. Thus, the total expected cost of the agent under the information $\mathcal{Y}_t$ can be expressed as
\begin{equation}\label{Vt}
\begin{split}
U_t &= \mathbb{E}\left[\int_0^T e^{-rt} f_A\big(t,p_t,E_t\big) dt + e^{-rT}h_A(M_T)\big| \mathcal{Y}_t,E_t = E^*_t\right]\\
&= \int_0^t e^{-rs} f_A\big(s,p_s,E_s^*\big) ds + e^{-rt}W_t.
\end{split}
\end{equation}
We further have conditions $U_0=W_0 = \underline{J}_A$ and $W_T = h_A(M_T)$. The effort $E_t = E^*_t$ indicates that the agent behaves optimally under a given contract.

\begin{proposition}\label{lem_martingale}
The total expected cost of the agent, $U_t$, is a martingale under $\mathcal{Y}_t$. In addition, there exists an $N$-dimensional progressively measureable process $\zeta_t$ such that 
\begin{equation}\label{dV}
dU_t = e^{-rt}\zeta_t^{\mathsf{T}} \left( dY_t - AY_tdt + E_t^*dt\right),
\end{equation}
where $\mathsf{T}$ denotes the transpose operator.
\end{proposition}
\begin{proof}
First, we have 
\begin{equation}
\begin{split}
\mathbb{E}[U_t|\mathcal{Y}_\tau]=&\mathbb{E}\left[\int_0^\tau e^{-rs} f_A(s,p_s,E_s^*) ds + e^{-r\tau}W_\tau\big|\mathcal{Y}_\tau\right]\\
&+\mathbb{E}\left[\int_\tau^t e^{-rs} f_A(s,p_s,E_s^*) ds + e^{-rt}W_t-e^{-r\tau}W_\tau\big|\mathcal{Y}_\tau\right]\\
=&U_\tau+\mathbb{E}\left[\int_\tau^t e^{-rs} f_A(s,p_s,E_s^*) ds + e^{-rt}W_t\big|\mathcal{Y}_\tau\right]-e^{-r\tau}W_\tau.
\end{split}
\end{equation}
Then, using \eqref{Wt}, we obtain 
\begin{equation}
\begin{split}
\mathbb{E}&\left[\int_\tau^t e^{-rs} f_A(s,p_s,E_s^*) ds + e^{-rt}W_t\big|\mathcal{Y}_\tau\right]  \\
&= \mathbb{E}\left[\int_\tau^T e^{-rs} f_A(s,p_s,E_s^*) ds +e^{-rT}h_A(M_T)\big|\mathcal{Y}_\tau\right] = e^{-r\tau}W_\tau.
\end{split}
\end{equation}
 Hence, $\mathbb{E}[U_t|\mathcal{Y}_\tau] = U_\tau$, and $U_t$ is a $\mathcal{Y}_t$-measurable martingale. Using martingale representation theorem \cite{karatzas2012brownian} yields \eqref{dV}. \qed
\end{proof}

Based on Proposition \ref{lem_martingale}, we can subsequently obtain the following lemma which facilitates design of the terminal payment term design in the optimal contract.
\begin{lemma}\label{lem:dM_1}
The aggregate equivalent income process $M_t$ evolves according to:
\begin{equation}\label{dM_1}
\begin{split}
dM_t = &\frac{rh_A(M_t)}{h_A'(M_t)}dt - \frac{f_A(t,p_t,E_t^*)}{h_A'(M_t)}dt+ \frac{1}{h_A'(M_t)}\zeta_t^{\mathsf{T}}( dY_t- AY_tdt+ E_t^*dt) \\
&  -\frac{1}{2}\frac{h_A''(M_t)}{h_A'(M_t)}\frac{\zeta_t^\mathsf{T}\Sigma_t(Y_t)\Sigma_t(Y_t)^\mathsf{T}\zeta_t}{h_A'^2(M_t)}dt.
\end{split}
\end{equation}
\end{lemma}

\begin{proof}
By substituting \eqref{dV} into \eqref{Vt}, we obtain
\begin{align}
&dU_t = e^{-rt} f_A\big(t,p_t,E_t^*\big)dt - re^{-rt}W_tdt + e^{-rt}dW_t,\notag\\
\Rightarrow\quad & dW_t = rW_tdt - f_A\big(t,p_t,E_t^*\big)dt + \zeta_t^{\mathsf{T}} \left( dY_t - AY_tdt + E_t^*dt\right).\label{dW}
\end{align}

Since $W_T = h_A(M_T)$, we adopt the form $W_t = h_A(M_t)$ and aim to characterize the contract that yields this form. Then, we have $\underline{J}_A = h_A(M_0) = h_A(c_0)$. Further, \eqref{dW} indicates that 
\begin{equation}\label{vol_compute}
\begin{split}
h_A'(M_t)dM_t + \frac{1}{2}h_A''(M_t)\chi_t^2 dt=& rh_A(M_t)dt- f_A\big(t,p_t,E_t^*\big)dt \\ &
+ \zeta_t^{\mathsf{T}} \left( dY_t - AY_tdt + E_t^*dt\right),
\end{split}
\end{equation}
where $\chi_t$ is the volatility of process $M_t$. Matching the volatility terms in \eqref{vol_compute} gives
$h_A'^2(M_t)\chi_t^2 = \zeta_t^\mathsf{T}\Sigma_t(Y_t)\Sigma_t(Y_t)^\mathsf{T}\zeta_t$.  
%(use the fact $adB_1+bdB_2=\sqrt{a^2+b^2}dB$ for independent BMs).
Then, \eqref{vol_compute} yields the result. \qed
\end{proof}

\textit{Remark:} Note that \eqref{dV} includes information on the cyber risk dynamics  \eqref{cyber_risk}. Thus, \eqref{dM_1} can be seen as a modified stochastic dynamic system of the agent with $M_t$ as a new state variable. In addition, $\zeta_t$ can be interpreted as the principal's control over the agent's revenue.

Another point to be highlighted is the role of $p_t$ in \eqref{dM_1}. Here, $p_t$ is not optimal yet and its value needs to be further determined by the principal. Currently, we can view $p_t$ as an exogenous variable that enters the constructed dynamic contract form \eqref{dM_1}. In addition, the feedback structure of the dynamic contract on $Y_t$ is reflected by the cumulative payment term $c_t$ shown later in Lemma \ref{lem:dc}.

\textit{Interpretation of Dynamic Contract:} The dynamic contract determines the risk manager's revenue in \eqref{dM_1}, which includes four separate terms. The first term, $\frac{rh_A(M_t)}{h_A'(M_t)}dt $, indicates that the risk manager's payoff should be increased to compensate the discounted future revenue. The second term, $- \frac{f_A(t,p_t,E_t^*)}{h_A'(M_t)}dt$, is an offset of the \textit{direct cost} of agent's effort. The third part, $\frac{1}{h_A'(M_t)}\zeta_t^{\mathsf{T}} \left( dY_t - AY_tdt + E_t^*dt\right)$, is an \textit{incentive} term, which captures the agent's benefit from spending effort in risk management. Here, the agent's real effort enters into the $Y_t$ term. The last one, $-\frac{1}{2}\frac{h_A''(M_t)}{h_A'(M_t)}\frac{\zeta_t^\mathsf{T}\Sigma_t(Y_t)\Sigma_t(Y_t)^\mathsf{T}\zeta_t}{h_A'^2(M_t)}dt$, is a \textit{risk compensation} term (the manager is risk-averse), capturing the fact that the risk manager faces uncertainties in the performance outcome due to the Brownian motion.

For completeness, we present the cumulative payment process $c_t$ in the following lemma.

\begin{lemma}\label{lem:dc}
The cumulative payment process $c_t$ evolves according to:
\begin{equation}\label{dc}
\begin{split}
dc_t =& \frac{rh_A(M_t)}{h_A'(M_t)}dt - \frac{f_A(t,p_t,E_t^*)}{h_A'(M_t)}dt+ \frac{1}{h_A'(M_t)}\zeta_t^{\mathsf{T}}( dY_t - AY_tdt + E_t^*dt)\\
&-\frac{1}{2}\frac{h_A''(M_t)}{h_A'(M_t)}\frac{\zeta_t^\mathsf{T}\Sigma_t(Y_t)\Sigma_t(Y_t)^\mathsf{T}\zeta_t}{h_A'^2(M_t)}dt-p_tdt.
\end{split}
\end{equation}
\end{lemma}

\begin{proof}
The result can be directly obtained from \eqref{M_t} and Lemma \ref{lem:dM_1}. \qed
\end{proof}

Lemma \ref{lem:dc} characterizes the cumulative payment process $c_t$ with initial value $c_0$ given by $h_A(c_0) = \underline{J}_A$. We focus on the class of contracts in \eqref{dc}, and aim to determine the optimal variables ($\zeta_t$ and $p_t$) to minimize the principal's cost. Note that \eqref{dc} is adapted to the principal's information set $\mathcal{Y}_t$, since the principal observes $M_t$ and $Y_t$, determines $p_t$, $\zeta_t$, and anticipates $E_t^*$. In addition, this payment process is directly related to the actual effort that the agent adopts, captured by $dY_t$. The variable $\zeta_t$ can be further interpreted as the sensitivity (or gain) of contract payment to the risk difference under the agent's optimal and actual efforts. 
In addition, since $W_t = h_A(M_t)$, based on \eqref{Wt}, we obtain 
\begin{equation}\label{totoal_cost_agent}
\begin{split}
U_t &= \mathbb{E}\left[\int_0^T e^{-rt} f_A\big(t,p_t,E_t\big) dt + e^{-rT}h_A(M_T)\big| \mathcal{A}_t\right] \\
&= \int_0^t e^{-rs} f_A\big(s,p_s,E_s^*\big) ds + e^{-rt}h_A(M_t),
\end{split}
\end{equation}
where the conditional expectation on $\mathcal{A}_t$ admits the same value as that on $\mathcal{Y}_t$. Proposition \ref{lem_martingale} indicates that $U_t$ is a martingale. Then, the expected value of  $e^{-rt}h_A(M_t)$ in \eqref{totoal_cost_agent} is zero which confirms the zero expected future cost of the agent.

\subsection{Incentive Analysis of Cyber Risk Manager}\label{incentive_analysis}
Recall that the principal suggests an optimal effort process $E_t^*$ by assuming that the agent behaves optimally. However, the agent can determine his actual effort $E_t$ that minimizes the cost $J_A$ based on $\mathcal{A}_t$ which might not be the same as $E_t^*$ that the principal suggests. Thus, the next important problem for the principal is to determine an incentive compatible contract. To achieve this goal, the principal determines the process $\zeta_t$ and the payment $p_t$ strategically to control the agent's actual effort $E_t$.

Denote by $V_a(t,M_t)$ the agent's value function with terminal condition $V_a(T,M_T) = h_A(M_T)$.  The property of value function ensures that the risk management effort is optimal if it satisfies the following dynamic programming equation: $e^{-rt}V_a(t,M_t) = \min_{E_t} \mathbb{E}\left\{ \int_t^s e^{-ru} f_A(u,p_u,E_u)du +e^{-rs}V_a(s,M_s) \right\}$.
Then, using \eqref{cyber_risk}, \eqref{M_t}, and \eqref{dc}, the cyber risk manager's revenue can be expressed as:
\begin{equation}\label{dM_t}
\begin{split}
dM_t =& \frac{rh_A(M_t)}{h_A'(M_t)}dt - \frac{f_A\big(t,p_t,E_t^*\big)}{h_A'(M_t)}dt + \frac{1}{h_A'(M_t)}\zeta_t^{\mathsf{T}} \left( E_t^*-E_t\right)dt \\
&- \frac{1}{2}\frac{h_A''(M_t)}{h_A'(M_t)}\frac{\zeta_t^\mathsf{T}\Sigma_t(Y_t)\Sigma_t(Y_t)^\mathsf{T}\zeta_t}{h_A'^2(M_t)}dt + \frac{1}{h_A'(M_t)}\zeta_t^{\mathsf{T}} \Sigma_t(Y_t) dB_t.
\end{split}
\end{equation}

We rewrite the risk manager's problem as follows:
\begin{align*}
\mathrm{(O-A')}:\quad &\min_{E_t\in\mathcal{E},\ t\in[0,T]}\ J_A\left(\{E_t\}_{0\leq t\leq T};\{p_t\}_{0\leq t\leq T},c_T\right)\\
\mathrm{subject\ to}\  &\mathrm{the\ stochastic\ dynamics}\ \eqref{dM_t},\ \mathrm{and\ the\ payment\ process}\ \eqref{M_t}.
\end{align*}
The Hamilton-Jacobi-Bellman (HJB) equation associated with the stochastic optimal control problem $\mathrm{(O-A')}$ is
\begin{equation}\label{HJB_A}
\begin{split}
\min_{E_t}\Bigg[  \frac{1}{2}\frac{\partial^2 V_a}{\partial M_t^2}\left( \frac{1}{h_A'^2(M_t)}  \zeta_t^{\mathsf{T}} \Sigma_t(Y_t) \Sigma_t(Y_t)^{\mathsf{T}}\zeta_t\right) + \frac{\partial V_a}{\partial M_t}\bigg( \frac{rh_A(M_t)}{h_A'(M_t)}- \frac{f_A\big(t,p_t,E_t^*\big)}{h_A'(M_t)} \\
 + \frac{1}{h_A'(M_t)}\zeta_t^{\mathsf{T}} \left( E_t^*-E_t\right)
-\frac{1}{2}\frac{h_A''(M_t)}{h_A'(M_t)}\frac{\zeta_t^\mathsf{T}\Sigma_t(Y_t)\Sigma_t(Y_t)^\mathsf{T}\zeta_t}{h_A'^2(M_t)}  \bigg) + f_A(t,p_t,E_t) \Bigg]+\frac{\partial V_a}{\partial t} = rV_a,\\
V_a(T,M_T) = h_A(M_T).
\end{split}
\end{equation}

Based on the candidate value function $V_a(t,M_t) = h_A(M_t)$, the second-order condition of \eqref{HJB_A} is satisfied. Then, the optimal solution to $\mathrm{(O-A')}$ is
\begin{equation}\label{agent_optimal_E}
\begin{split}
E_t^o &= {\arg\max}_{E_t}\ \frac{\partial V_a}{\partial M_t}  \frac{1}{h_A'(M_t)}\zeta_t^{\mathsf{T}} E_t - f_A(t,p_t,E_t) \\
& = {\arg\max}_{E_t}\ \zeta_t^{\mathsf{T}} E_t -  f_A(t,p_t,E_t).
\end{split}
\end{equation}
For a given contract, $E_t^o$ is the optimal effort of the agent. Then, when the anticipated effort $E^*_t$ of the principal coincides with $E_t^o$, i.e., $E^*_t=E^o_t$, the provided contract is IC and $E^*_t$ is implemented. The following theorem captures this result.

\begin{theorem}\label{IC_thm}
When the compensation process in the contract is specified by \eqref{dc}, then the IC constraint is satisfied, i.e., $E_t^*$ is implemented as expected by the principal, if and only if the following condition holds:
\begin{equation}\label{IC_simplified}
E_t^* = {\arg\max}_{E_t}\ \zeta_t^{\mathsf{T}} E_t -  f_A(t,p_t,E_t),
\end{equation}
where $\zeta_t$ is adapted to the information $\mathcal{Y}_t$ available to the principal.
\end{theorem}

\begin{proof}
We verify that $E_t^*$ is implemented by the agent.

 For an arbitrary process $\{E_t\}_{0\leq t\leq T}$, we define a variable $$\tilde{U}_t=\int_0^t e^{-rs} f_A\big(s,p_s,E_s\big) ds + e^{-rt}h_A(M_t),$$
  where $M_t$ is given by \eqref{dM_t}. Note that the HJB equation associated with $\mathrm{(O-A')}$ can also be written as $0=\min_{E_t}\ \mathbb{E}\left[d\tilde{U}_t|\mathcal{A}_t\right].$ Then, we know that when $E_t\neq E_t^*$, the drift term of $\tilde{U}_t$ is positive and yields $\tilde{U}_t<\mathbb{E}[\tilde{U}_T|\mathcal{A}_t]$. Hence, at time $t$, the expected total cost of the risk manager is greater than $\tilde{U}_t$. When $E_t=E^*_t$, we have $\mathbb{E}\left[d\tilde{U}_t|\mathcal{A}_t\right]=0$, and thus $\tilde{U}_t=\mathbb{E}[\tilde{U}_T|\mathcal{A}_t]$. This verifies that $E_t^*$ is the incentive compatible optimal decision of the risk manager such that his total expected cost is achieved at the lower bound. \qed
\end{proof}

Based on Theorem \ref{IC_thm}, the principal can indirectly manipulate the implemented effort of the agent by determining the variables $\zeta_t$ and $p_t$ jointly. Hence, under \eqref{IC_simplified}, the suggested effort $E_t^*$ is incentive compatible. 

\textit{Remark:} From \eqref{IC_simplified}, we can see that the risk manager's behavior is \textit{strategically neutral}. Specifically, at time $t$, the risk manager decides on the optimal effort $E_t^*$ based only on the current cost (term $f_A(t,p_t,E_t)$) and benefit (term $\zeta_t^\mathsf{T} E_t$) instead of future-looking variables. This neutral behavior is consistent with the fact that a larger current effort does not induce a higher payoff for the agent after time $t$, since as shown in \eqref{totoal_cost_agent}, the expected future cost over time $(t,T]$ is zero due to the martingale property.

\section{The Principal's Problem: Optimal Dynamic Systemic Cyber Risk Management}\label{principal_problem}
Our next goal is to characterize the dynamic contracts designed by the principal. Furthermore, we present a separation principle and explicit solutions to an LQ case in this section.

\subsection{Rational Controllability}
The controllability of the cyber risk is critical to the principal. To account for the incentives in the management of risk, we have the following definition.

\begin{definition}[Rational Controllability]
The dynamic systemic cyber risk is rationally controllable if the principal can provide incentives $\{p_t\}_{0\leq t\leq T}$ and $c_T$ such that the risk manager's effort $\{E_t\}_{0\leq t\leq T}$ coincides with the one suggested by the principal.
\end{definition}

In ODMD, the rational controllability indicates that under $\{\{p^*_t\}_{0\leq t\leq T},c_T^*\}$, the best-response behavior $\{E_t^*\}_{0\leq t\leq T}$ of the agent is the same as the principal's predicted effort. The unique feature of rational controllability is that the principal cannot control the cyber risk directly but can rely on other terms to infer the rational behavior of the agent, which further influences the applied effort in risk management. Corollary \ref{coro1} later captures this result.
%Therefore, to address $\mathrm{(O-P)}$, we need to analyze the relationship between the compensations and the suggested effort. 

\subsection{Stochastic Optimal Control Reformulation}\label{principal_problem_reform}
Knowing that the cyber risk manager behaves strategically, the principal aims to implement $E_t^*$ and thus \eqref{dc} becomes
\begin{equation}\label{dc_simplified}
\begin{split}
dc_t =& \frac{rh_A(M_t)}{h_A'(M_t)}dt - \frac{f_A\big(t,p_t,E_t^*\big)}{h_A'(M_t)}dt-\frac{1}{2}\frac{h_A''(M_t)}{h_A'(M_t)}\frac{\zeta_t^\mathsf{T}\Sigma_t(Y_t)\Sigma_t(Y_t)^\mathsf{T}\zeta_t}{h_A'^2(M_t)}dt\\
& -p_tdt + \frac{1}{h_A'(M_t)}\zeta_t^{\mathsf{T}} \Sigma_t(Y_t)dB_t.
\end{split}
\end{equation}

Instead of dealing with the complex revenue dynamics \eqref{dM_t} of the principal, we deal with its equivalent counterpart $dh_t$ shown in Theorem \ref{P_problem_thm} below, which is much simpler. We reformulate the principal's problem as a standard stochastic optimal control problem as follows.

\begin{theorem}\label{P_problem_thm}
The principal's problem is reformulated as a stochastic optimal control problem as follows:
\begin{align*}
\mathrm{(O-P')}:\ \min_{p_t\in\mathcal{P},\ \zeta_t}\ \mathbb{E}&\int_0^T e^{-rt} \left( f_P(t,Y_t,p_t) - e^{-r(T-t)}p_t\right)dt \\
&\qquad\qquad +e^{-rT} \left(h_P(Y_T)+h_A^{-1}(h_T)\right)\\
\mathrm{such\ that}\quad & dY_t = AY_tdt - E_t^* dt + \Sigma_t(Y_t) dB_t,\ Y_0 = y_0,\\
&dh_t = rh_tdt - f_A(t,p_t,E_t^*)dt + \zeta_t^{\mathsf{T}} \Sigma_t(Y_t) dB_t,\ h_0 = \underline{J}_A,\\
&E_t^* = {\arg\max}_{E_t}\ \zeta_t^{\mathsf{T}} E_t -  f_A(t,p_t,E_t).
\end{align*}
\end{theorem}

\begin{proof}
Recall that the expected cost of the cyber risk manager is equal to $W_t = h_A(M_t)$. Then, under the optimal risk management effort and denoting $h_t = h_A(M_t)$, we obtain
$$
dh_t = rh_tdt - f_A(t,p_t,E_t^*)dt + \zeta_t^{\mathsf{T}} \Sigma_t(Y_t)dB_t,\ h_0 = \underline{J}_A.
$$
In addition, based on $dc_t = dM_t-p_tdt$, we have
$
c_T = M_T - \int_0^T p_t dt.
$
Since $M_T = h_A^{-1}(h_T)$, we have $e^{-rT}c_T = e^{-rT}h_A^{-1}(h_T) - e^{-rt}\int_0^T e^{-r(T-t)}p_t dt.$
Thus, the cost function of the principal can be rewritten as
$$
\mathbb{E}\int_0^T e^{-rt} \left( f_P(t,Y_t,p_t) - e^{-r(T-t)}p_t\right)dt 
+ e^{-rT}\left(h_P\big(Y_T)+h_A^{-1}(h_T)\right),
$$
which yields the result. \qed
\end{proof}

In the investigated incomplete information situations, the principal preserves the indirect controllability of systemic risk $Y_t$ by estimating the agent's effort $E_t^*$ as well as specifying the contract terms $p_t,\ c_T$ and process $\zeta_t$.

\begin{corollary}\label{coro1}
By providing incentives $\{\{p_t\}_{0\leq t\leq T},c_T\}$ and specifying process $\{\zeta_t\}_{0\leq t\leq T}$, the dynamic systemic cyber risk is rationally controllable, and the incentive compatible effort follows \eqref{IC_simplified}. The optimal $\{p^*_t\}_{0\leq t\leq T}$ and $\{\zeta^*_t\}_{0\leq t\leq T}$ can be obtained from Theorem \ref{P_problem_thm}.
\end{corollary}
\begin{proof}
The result directly follows from Theorems \ref{IC_thm} and \ref{P_problem_thm}. \qed
\end{proof}

\textit{Remark:} Theorem \ref{P_problem_thm} presents solution to a standard optimal control problem for the principal, whose the existence and uniqueness have been well studied \cite{yong1999stochastic}. With  $f_P$, $h_P$, $f_A$, and $h_A$ satisfying the conditions in Assumptions \ref{Assump_1} and \ref{Assump_2}, and the corresponding coefficients in the functions well selected ensuring the feasibility of $\mathrm{(O-P')}$, the control problem can be solved efficiently by numerical methods \cite{kushner1990numerical}. Therefore,  the ODMD for the systemic risk management problem, i.e., $E_t^*$, $p_t^*$, and $c_T^*$, can be determined  from \eqref{IC_simplified}, \eqref{dc_simplified}  and Theorem \ref{P_problem_thm}, respectively.

\subsection{Separation Principle} 
We next present a separation principle for the asset owner in determining the compensation $p_t$ and the auxiliary parameter $\zeta_t$.
First, we make assumptions on the separability of the cost functions.

\textbf{(S1)}: The agent's running cost can generally be separated into two parts, including the effort and payment.  Accordingly, we take $f_A(t,p_t,E_t)$ to be in the form 
\begin{equation}
f_A(t,p_t,E_t) = f_{A,E}(E_t) - f_{A,p}(p_t),
\end{equation}
where $f_{A,E}:\mathcal{E}\rightarrow \mathbb{R}_+$ is monotonically increasing, continuously differentiable and strictly convex, i.e., $f_{A,E}'(E_t)>0$ and $f_{A,E}''(E_t)>0$, and $f_{A,p}:\mathcal{P}\rightarrow \mathbb{R}_+$. Then, the constraint $E_t^* = {\arg\max}_{E_t}\ \zeta_t^{\mathsf{T}} E_t -  f_A(t,p_t,E_t)$ can be simplified to 
\begin{equation}
E_t^*  = f_{A,E}'^{-1}(\zeta_t).
\end{equation}

\textbf{(S2)}: We also assume that the principal's running cost takes the form 
\begin{equation}
f_P(t,Y_t,p_t) = f_{P,Y}(Y_t) + f_{P,p}(p_t),
\end{equation}
where $f_{P,Y}:\mathbb{R}^N\rightarrow \mathbb{R}$ and $f_{P,p}:\mathcal{P}\rightarrow \mathbb{R}_+$ are monotonically increasing and continuously differentiable.

The inverse function $h_A^{-1}$ plays a role in the principal's objective. We further have the following assumption.

\textbf{(L1)}: The agent's terminal cost function $h_A$ is linear, i.e., $h_A(M_T) = \gamma M_T$, where $\gamma<0$.

Then, we have the following \textit{separation principle}.
\begin{theorem}\label{separation_p}
Under conditions \textbf{(S1)}, \textbf{(S2)}, and \textbf{(L1)}, the principal's problem $\mathrm{(O-P')}$ can be separated into two subproblems with respect to the decision variables $\zeta_t$ and $p_t$ as:
\begin{align*}
(SP1):\ \min_{\zeta_t}&\ \mathbb{E}\ \int_0^T e^{-rt}  \left(f_{P,Y}(Y_t)-\frac{1}{\gamma}f_{A,E}\left(f_{A,E}'^{-1}(\zeta_t)\right)\right) dt\\
&\qquad+e^{-rT}h_P(Y_T)+ \frac{1}{\gamma} \int_{0}^T e^{-rt} \zeta_t^{\mathsf{T}} \Sigma_t(Y_t)dB_t \\
\mathrm{such\ that}\quad  dY_t& = AY_tdt - f_{A,E}'^{-1}(\zeta_t) dt + \Sigma_t(Y_t) dB_t,\ Y_0 = y_0.\\
(SP2):\ \min_{p_t\in\mathcal{P}}\ &\int_0^T e^{-rt} \left( f_{P,p}(p_t) - e^{-r(T-t)}p_t+\frac{1}{\gamma}f_{A,p}(p_t) \right)dt.
\end{align*}
\end{theorem}
\begin{proof}
For the constraint $dh_t = rh_tdt - f_{A,E}(f_{A,E}'^{-1}(\zeta_t))dt+f_{A,p}(p_t)dt + \zeta_t^{\mathsf{T}} \Sigma_t(Y_t)dB_t$, we obtain
$
h_t = e^{rt}h_0-\int_{0}^t e^{r(t-s)} [f_{A,E}\big(f_{A,E}'^{-1}(\zeta_s)\big)-f_{A,p}(p_s)]ds 
+ \int_{0}^t e^{r(t-s)} \zeta_s^{\mathsf{T}} \Sigma_s(Y_s)dB_s.
$
Thus, the principal's problem can be rewritten as
\begin{align*}
\min_{p_t\in\mathcal{P},\zeta_t}\ \mathbb{E}&\int_0^T e^{-rt} \left( f_{P,Y}(Y_t) + f_{P,p}(p_t) - e^{-r(T-t)}p_t\right)dt\\
&+ e^{-rT} \Big[h_P(Y_T)+h_A^{-1}\big(e^{rT}\underline{J}_A-\int_{0}^T e^{r(T-s)}f_{A,E}\big(f_{A,E}'^{-1}(\zeta_s)\big)ds \\
&+\int_{0}^T e^{r(T-s)}f_{A,p}(p_s)ds + \int_{0}^T e^{r(T-s)} \zeta_s^{\mathsf{T}} \Sigma_s(Y_s)dB_s\big) \Big]\\
\mathrm{such\ that}&\quad  dY_t = AY_tdt - f_{A,E}'^{-1}(\zeta_t) dt + \Sigma_t(Y_t) dB_t,\ Y_0 = y_0.
\end{align*}
Then, the decomposition of the problem follows naturally. \qed
\end{proof}

\textit{Remark:} $\zeta_t$ can be regarded as an \textit{estimation variable} since it determines the anticipated effort $E_t^*$. The payment $p_t$ is a \textit{control variable} that manipulates the risk manager's incentives and is determined at the control phase. Under appropriate conditions, these two estimation and control variables can be designed in a separate manner, yielding a separation principle in dynamic contract design for systemic risk management. 

To obtain more insights, we next focus on a class of models where the value function of the principal and the ODMD can be explicitly characterized.

\subsection{ODMD in LQ Setting}\label{LQ_case_section}

In the LQ setting, the cost functions take forms as $f_{A,E}(E_t) = \frac{1}{2}E_t^{\mathsf{T}}R_tE_t$, and $f_{A,p}(p_t) = \delta_A p_t$, where  $R_t$ is a positive-definite $N\times N$-dimensional symmetric matrix and $\delta_A$ is a positive constant. Then we obtain 
\begin{equation}
E_t^* = f_{A,E}'^{-1}(\zeta_t) = R_t^{-1}\zeta_t.
\end{equation}
Further, we consider $h_P(Y_T) = \rho^{\mathsf{T}}Y_T$, where $\rho\in\mathbb{R}_{+}^N$ maps the cyber risks to monetary loss, and $f_P(t,Y_t,p_t)=\rho^{\mathsf{T}}Y_t+\delta_P p_t$, where $\delta_P$ is a positive constant. In addition, $h_A(M_T) = -M_T$ and $\Sigma_t(Y_t) = D_t \cdot diag(Y_t)$, where $D_t\in\mathbb{R}^{N\times N}$ and `$diag$' is a diagonal operator. The principal's problem becomes:
\begin{align*}
\min_{p_t\in\mathcal{P},\zeta_t}\ \mathbb{E}&\int_0^T e^{-rt} ( \rho^{\mathsf{T}}Y_t+\delta_P p_t - e^{-r(T-t)}p_t)dt + e^{-rT}(\rho^{\mathsf{T}}Y_T-h_T)\\
\mathrm{such\ that}\quad &  dY_t = (AY_t - R_t^{-1}\zeta_t) dt + D_t\cdot diag(Y_t) dB_t,\ Y_0 = y_0,\\
&dh_t = \left(rh_t - \frac{1}{2}\zeta_t^{\mathsf{T}} R_t^{-1} \zeta_t+\delta_A p_t\right)dt + \zeta_t^{\mathsf{T}} \Sigma_t(Y_t)dB_t,\ h_0 = \underline{J}_A.
\end{align*}

The principal aims to maximize $h_T$, which is equivalent to minimizing the agent's total revenue based on the relationship $h_T=-M_T$. The principal also considers the agent's participation constraint by setting $h_0=W_0=\underline{J}_A$, ensuring that the cyber risk manager has sufficient incentive to fulfill the task.

Since $
e^{-rT}h_T = h_0-\int_{0}^T e^{-rs} \left(\frac{1}{2}\zeta_s^{\mathsf{T}} R_s^{-1} \zeta_t-\delta_A p_s \right) ds + \int_{0}^T e^{-rs} \zeta_s^{\mathsf{T}} D_s\cdot diag(Y_s)dB_s
$, the principal's problem can be rewritten as:
\begin{align*}
\min_{p_t\in\mathcal{P},\zeta_t}\ \mathbb{E}&\int_0^T e^{-rt} \Big( \rho^{\mathsf{T}}Y_t + (\delta_P-\delta_A)p_t - e^{-r(T-t)}p_t+\frac{1}{2}\zeta_t^{\mathsf{T}} R_t^{-1} \zeta_t\Big)dt\\
&\qquad\qquad\qquad\qquad\qquad\qquad\qquad\qquad + e^{-rT} \rho^{\mathsf{T}}Y_T-\underline{J}_A\\
\mathrm{such\ that}\quad & dY_t = (AY_t - R_t^{-1}\zeta_t) dt + D_t\cdot diag(Y_t) dB_t,\ Y_0 = y_0.
\end{align*}
According to Theorem \ref{separation_p}, the separation principle holds in the LQ case. To determine the optimal $p_t$, we solve the following unconstrained optimization problem:
$$\min_{p_t\in\mathcal{P}}\ \int_0^T e^{-rt}(\delta_P-\delta_A- e^{-r(T-t)})p_tdt.$$ 
Depending on the values of parameters $\delta_P$ and $\delta_A$, we obtain the following results. If $\delta_P-\delta_A\geq 1$, there is no intermediate payment, i.e., $p_t=0$, $\forall t\in[0,T]$. In this regime, the principal has a higher valuation on the monetary payment than the agent does. In other words, the agent is relatively hard to be incentivized to do the risk management. When $\delta_P-\delta_A\leq 0$, i.e., the principal focuses more on the cyber risk deduction rather than the expenditure on  incentivizing the agent,  the optimal $p_t$ is positively unbounded. However, in this regime, the terminal payment $c_T$ is negatively unbounded based on \eqref{dc_simplified}. This contract corresponds to the scenario where the risk manager receives a large amount of intermediate payment during the task while returning it to the principal after finishing the task which is not practical. Under $0<\delta_P-\delta_A< 1$, the intermediate compensation is either 0 or unbounded depending on the time index.
Hence, to design a practical contract, we focus on the regime in which the intermediate payment is zero, and the risk manager receives a positive terminal payment $c_T$.

To obtain the optimal $\{\zeta^*_t\}_{0\leq t\leq T}$, we assume that the process $\zeta_t$, ${t\in[0,T]}$, is non-anticipative, which can be verified later after obtaining the solution $\zeta_t^*$. Then, the problem can be further simplified to:
\begin{align*}
\min_{\zeta_t}\ \mathbb{E}&\int_0^T e^{-rt}\left( \rho^{\mathsf{T}}Y_t + \frac{1}{2}\zeta_t^{\mathsf{T}} R_t^{-1} \zeta_t\right) dt+ e^{-rT}\rho^{\mathsf{T}} Y_T - \underline{J}_A\\
\mathrm{such\ that}\quad & dY_t = (AY_t - R_t^{-1}\zeta_t) dt + D_t\cdot diag(Y_t) dB_t,\ Y_0 = y_0.
\end{align*}

The following theorem provides the optimal solution $\zeta_t^*$.
\begin{theorem}\label{optimal_zeta_LQ}
In the LQ case, the optimal solution to the principal's problem is given by 
\begin{equation}
\zeta_t^* = K_t,
\end{equation} where $K_t$ satisfies, and is the unique solution to 
\begin{equation}\label{sol_K}
\dot{K}_t+(A-rI)^{\mathsf{T}}K_t+\rho = 0,\ K_T = \rho. 
\end{equation}
Furthermore, the minimum cost of the principal is given by 
\begin{equation}
J_p^* = K_0^{\mathsf{T}}y_0+m_0 - \underline{J}_A,
\end{equation}
where $m_0$ is obtained uniquely from 
\begin{equation}\label{sol_m}
\dot{m}_t-rm_t-\frac{1}{2}K_t^{\mathsf{T}}R_t^{-1}K_t=0,\ m_T=0.
\end{equation}
\end{theorem}

\begin{proof}
Without loss of generality, we solve the optimal control problem by ignoring the constant term $\underline{J}_A$ in the cost function.
The HJB equation
\begin{equation}
\begin{split}
\min_{\zeta_t}\Big[  \frac{1}{2}tr \left(\frac{\partial^2 V_p}{\partial Y_t^2}D_t\cdot diag(Y_t) \cdot diag(Y_t) D_t^{\mathsf{T}} \right)
+ \frac{\partial V_p}{\partial Y_t}\left( AY_t-R_t^{-1}\zeta_t \right) \\
+ \rho^{\mathsf{T}}Y_t + \frac{1}{2}\zeta_t^{\mathsf{T}} R_t^{-1} \zeta_t \Big]+\frac{\partial V_p}{\partial t} = rV_p,\\
V_p(T,Y_T) = \rho^{\mathsf{T}} Y_T,
\end{split}
\end{equation}
yields the  first-order condition $\zeta_t^*=\frac{\partial V_p}{\partial Y_t}$. Assume that the value function takes the form: $V_p(t,Y) = \frac{1}{2}Y^{\mathsf{T}}S_tY+K_t^{\mathsf{T}}Y+m_t$, where $S_t$ is an $N\times N$ symmetric matrix with continuously differentiable entries, $K_t$ is a continuously differentiable $N$-dimensional vector, and $m_t$ is a continuously differentiable function. Then, we obtain $\zeta_t^* = S_tY_t+K_t$. Substituting $\zeta_t^*$ into the HJB equation yields
\begin{equation}\label{HJB_LQR}
\begin{split}
 \frac{1}{2}tr \left(S_tD_t\cdot diag(Y_t) \cdot diag(Y_t) D_t^{\mathsf{T}}\right)+ (S_tY_t+K_t)^{\mathsf{T}} ( AY_t
 -R_t^{-1}S_tY_t-R_t^{-1}K_t )\\ + \rho^{\mathsf{T}}Y_t + \frac{1}{2}(S_tY_t+K_t)^{\mathsf{T}}R_t^{-1}(S_tY_t+K_t)  \\
 = r\left(\frac{1}{2}Y_t^{\mathsf{T}}S_tY_t+K_t^{\mathsf{T}}Y_t+m_t\right)-\frac{1}{2}Y_t^{\mathsf{T}}\dot{S}_tY_t-\dot{K}_t^{\mathsf{T}}Y_t-\dot{m}_t,\\
V_p(T,Y_T) = \rho^{\mathsf{T}} Y_T.
\end{split}
\end{equation}
Denote by $I$ the $N$-dimensional identity matrix and by $e_i$ the $N$-dimensional vector whose $i$-th element is 1 and the others are zero. Matching the coefficients in \eqref{HJB_LQR} further yields the following coupled ordinary differential equations (ODEs):
\begin{align}
\dot{S}_t+S_tA+A^{\mathsf{T}}S_t-rS_t-S_tR_t^{-1}S_t+\frac{1}{2}\sum_{i=1}^N\left(e_i e_i^{\mathsf{T}}D_t^{\mathsf{T}} S_tD_t\right)&=0,\ S_T=0,\label{sol_S}\\
\dot{K}_t+(A-R_t^{-1}S_t-rI)^{\mathsf{T}}K_t+\rho = 0,&\ K_T = \rho,\label{sol_K_ax}\\
\dot{m}_t-rm_t-\frac{1}{2}K_t^{\mathsf{T}}R_t^{-1}K_t=0,&\ m_T=0\label{sol_m_ax}.
\end{align}
Here, \eqref{sol_S} is a matrix Riccati equation. However, based on the terminal condition $S_T=0$, we see that the unique solution to \eqref{sol_S} is $S_t=0$, $\forall t$. Therefore, a linear value function $V_p(t,Y)=K_t^{\mathsf{T}}Y+m_t$ is sufficient. Then, the ODEs \eqref{sol_K_ax} and \eqref{sol_m_ax} can be rewritten as \eqref{sol_K} and \eqref{sol_m}, respectively, which being linear admit unique solutions. \qed
\end{proof}

We then obtain the explicit form of optimal dynamic contract in the subsequent lemma.
\begin{lemma}\label{lem_2}
In the LQ case, the optimal dynamic contract designed by the principal is given by
\begin{equation}\label{dc_lem_2}
\begin{split}
dc_t &= \left( rc_t+ \frac{1}{2}K_t^{\mathsf{T}}R_t^{-1}K_t \right) dt -K_t^{\mathsf{T}} \left( dY_t - AY_tdt + R_t^{-1}K_t dt\right)\\
& = \left( rc_t- \frac{1}{2}K_t^{\mathsf{T}}R_t^{-1}K_t \right) dt -K_t^{\mathsf{T}} \left( dY_t - AY_tdt\right),
\end{split}
\end{equation}
with $c_0=-\underline{J}_A>0$, and $K_t$ is given by \eqref{sol_K}. The intermediate payment $p_t$ degenerates to zero, and the anticipated effort of the agent under the optimal contract is $E_t^* = R_t^{-1}K_t$.
\end{lemma}
\begin{proof}
The result follows from Theorems \ref{IC_thm}, \ref{optimal_zeta_LQ}, and \eqref{dc_simplified}. \qed
\end{proof}
\textit{Remark:} As shown in Lemma \ref{lem_2}, the cyber risk volatility $\Sigma_t(Y_t)$ does not impact the optimal dynamic contract design, since the principal's expected cost is linear in the systemic risk $Y_t$. When one of the functions $f_p$, $h_A$ and $h_p$ is not linear, the volatility $\Sigma_t(Y_t)$ will play a role in the contract design in solving the problem presented in Theorem \ref{P_problem_thm}.

Even though the optimal dynamic contract does not depend on the cyber risk volatility in the LQ case, the risk volatility influences the real compensation during contract implementation.

\begin{corollary}\label{volatility_impact}
The terminal compensation of risk manager has a larger variance when there are more complex interdependencies of risk uncertainties between nodes.
\end{corollary}

Corollary \ref{volatility_impact} will further be illustrated through case studies in Section \ref{examples}.

\section{Benchmark Scenario: Systemic Cyber Risk Management under Full Information }\label{benchmark}
In the full-information case, the principal observes the efforts that the cyber risk manager implements. We first solve the team problem in which the agent cooperates with the principal. To that end, the principal's cost under the team optimal solution is the best that he can achieve. Then, we aim to design a dynamic contract mechanism under which the agent will adopt the same policy as the team optimal one. In the cooperative case, the contract only needs to guarantee the participation constraint. Then, the principal's problem can be formulated as follows:
\begin{align*}
\mathrm{(O-B)}:\
 \min_{p_t\in\mathcal{P},c_T,E_t\in\mathcal{E}}&\ \mathbb{E}\int_0^T e^{-rt} f_P(t,Y_t,p_t)dt + e^{-rT}\left( c_T+ h_P(Y_T)\right)\\
\mathrm{such\ that}\quad    & dY_t = AY_tdt - E_t dt + \Sigma_t(Y_t) dB_t,\ Y_0 = y_0,\\
& J_A\left(\{E_t^*\}_{0\leq t\leq T};\{p_t\}_{0\leq t\leq T},c_T\right) = \underline J_A.
\end{align*} 

As in the asymmetric information scenario, it is more convenient to deal with the dynamics of the cyber risk manager's expected cost. By designing the contract, the principal only needs to ensure the participation of the agent. Then, the principal's problem can be rewritten as follows:
\begin{align*}
\mathrm{(O-B')}:\
 \min_{p_t\in\mathcal{P},\zeta_t,E_t\in\mathcal{E}}\ &\mathbb{E}\int_0^T e^{-rt}  \Big( f_P(t,Y_t,p_t) - e^{-r(T-t)}p_t\Big) dt \\
&\qquad\qquad+ e^{-rT} \left(h_P\big(Y_T\big)+h_A^{-1}(h_T)\right)\\
\mathrm{such\ that}\quad  dY_t =& AY_tdt - E_t dt + \Sigma_t(Y_t) dB_t,\ Y_0 = y_0,\\
dh_t =& rh_tdt - f_A\big(t,p_t,E_t\big)dt + \zeta_t^{\mathsf{T}} \Sigma_t(Y_t)dB_t,\ h_0 = \underline{J}_A.
\end{align*} 
With the full observation of $Y_t$ and $E_t$, $\zeta_t$ can be chosen freely, and $E_t$ can be seen as a control variable of the principal. Note that the IC constraint \eqref{IC_simplified} does not enter into $\mathrm{(O-B')}$. In addition, the equivalent terminal payment process $c_t$ admits the same form as \eqref{dc_simplified}. $\mathrm{(O-B')}$ is a standard stochastic optimal control problem which can be solved efficiently.

To quantify the efficiency of dynamic contract designed in Section \ref{principal_problem}, we have the following definition.

\begin{definition}[Information Rent] Denote the solutions to $\mathrm{(O-A)}$ and $\mathrm{(O-P)}$ by $\{E_t^*\}_{0\leq t\leq T}$ and $\{\{p_t^*\}_{0\leq t\leq T},c_T^*\}$, respectively. Further, denote the solution to $\mathrm{(O-B)}$ by $\{\{p_t^b\}_{0\leq t\leq T},c_T^b,\{E_t^b\}_{0\leq t\leq T}\}$. Then, the information rent is given by
\begin{equation}
I_R = J_P(\{p_t^*\}_{0\leq t\leq T},c_T^*)-J_P(\{p_t^b\}_{0\leq t\leq T},c_T^b).
\end{equation}
\end{definition}

Intuitively, information rent quantifies the difference between the principal's costs with optimal mechanisms designed under incomplete and full information.

We have following result on information rent.

\begin{corollary}
The optimal cost of the principal under full information is no larger than the one under asymmetric information. Hence, $I_R\geq 0$. 
\end{corollary}
\begin{proof}
Comparing with the optimal $\{E_t^*\}_{0\leq t\leq T}$ in $\mathrm{(O-P')}$, the implemented effort $\{E_t^b\}_{0\leq t\leq T}$ in $\mathrm{(O-B')}$ does not depend on the variables $\zeta_t$ and $p_t$. Thus, $\mathrm{(O-B')}$ admits a larger feasible solution space, which yields the result. \qed
\end{proof}

\subsection{LQ Setting: Certainty Equivalence Principle}
To further characterize the optimal contracts under full information and quantify the information rent, we investigate a class of special scenarios. Specifically,
we take the functions to have the same forms as in Section \ref{LQ_case_section}. The principal's problem can then be written as
\begin{align*}
\min_{p_t\in\mathcal{P},E_t\in\mathcal{E}}\quad \mathbb{E}&\int_0^T e^{-rt} \Big( \rho^{\mathsf{T}}Y_t + (\delta_P-\delta_A)p_t - e^{-r(T-t)}p_t+\frac{1}{2}E_t^{\mathsf{T}} R_t E_t\Big)dt\\
&\qquad\qquad\qquad\qquad\qquad\qquad\qquad\qquad + e^{-rT} \rho^{\mathsf{T}}Y_T-\underline{J}_A\\
\mathrm{such\ that}\quad & dY_t = (AY_t - E_t) dt + D_t\cdot diag(Y_t) dB_t,\ Y_0 = y_0.
\end{align*}
Note that $\zeta_t$ does not appear in the optimization problem. However, $\zeta_t$ enters the designed contract \eqref{dc_simplified} through the term $-\zeta_t^{\mathsf{T}} \Sigma_t(Y_t)dB_t$. In the long term contracting when $T$ is relatively large, the expected value of $-\zeta_t^{\mathsf{T}} \Sigma_t(Y_t)dB_t$ is zero which is irrelevant with $\zeta_t$. Hence, the principal can set $\zeta_t=0$ to reduce the contract complexity.

Similar to the analysis in Section \ref{LQ_case_section}, we focus on the regime where the intermediate payment flow $p_t$ is zero, to avoid the unrealistic situation of negative terminal payment. We obtain the following lemma characterizing the \textit{certainty equivalence principle}.

\begin{lemma}
In the LQ settings, $I_R = 0$ which reveals the certainty equivalence principle, i.e., the designed optimal contracts under the incomplete information are as efficient as those designed under complete information. 
\end{lemma}
\begin{proof}
By regarding $E_t$ as the role of $R_t^{-1}\zeta_t$, we see that the problem is reduced to the one in Section \ref{LQ_case_section}. Hence, the minimum cost of the principal in the full information case is the same as that under the incomplete information. \qed
\end{proof}

\textit{Remark:} When the agent's terminal cost function $h_A$ is not linear, $h_A^{-1}(h_T)$ will not be linear in $h_T$. Thus, the decision variable $\zeta_t$ remains in the principal's objective function. Then, the contract design under full information becomes more efficient as there is no dependency between $\zeta_t$ and $E_t$ introduced by the IC constraint.

In the LQ case, the team optimal contract is summarized as follows.
\begin{lemma}\label{lem_LQ_team}
In the LQ setting, the team optimal dynamic contract  is
\begin{equation}\label{lem_LQ_team_ct}
\begin{split}
dc_t^b &= \left( rc_t^b+ \frac{1}{2}K_t^{\mathsf{T}}R_t^{-1}K_t \right) dt,\\
E_t^b &= R_t^{-1}K_t,
\end{split}
\end{equation}
with $c_0^b=-\underline{J}_A>0$, and $K_t$ is given by \eqref{sol_K}. The intermediate payment is zero.
\end{lemma}
\begin{proof}
The result follows immediately from Theorem \ref{optimal_zeta_LQ} and \eqref{dc_simplified} with $\zeta_t=0$. \qed
\end{proof}

The following lemma provides a mechanism that leads to implementation of the team optimal solution presented in Lemma \ref{lem_LQ_team} without forcing the agent to follow $E_t^b$.

\begin{lemma}
In the LQ setting, the implementable optimal dynamic contract designed by the principal under full information is
\begin{equation}
dc_t = \left( rc_t- \frac{1}{2}K_t^{\mathsf{T}}R_t^{-1}K_t + K_t^{\mathsf{T}}E_t\right) dt,
\end{equation}
with $c_0=-\underline{J}_A>0$ and $K_t$ given by \eqref{sol_K}. The intermediate payment is zero, and the agent's best response is $E_t = R_t^{-1}K_t$.
\end{lemma}

\begin{proof}
Similar to the methodologies proposed in \cite{bacsar1984affine,cansever1985stochastic,bacsar1989stochastic}, we let the contract take the following form:
\begin{equation}\label{implementable_ct}
dc_t = \left( rc_t+ \frac{1}{2}K_t^{\mathsf{T}}R_t^{-1}K_t \right) dt + \Gamma_t^{\mathsf{T}}(E_t-R_t^{-1}K_t)dt,
\end{equation}
where $\Gamma_t$ is an $N$-dimensional vector to be determined. The second term $\Gamma_t^{\mathsf{T}}(E_t-R_t^{-1}K_t)dt$ is introduced to penalize the agent when his action deviates from $R_t^{-1}K_t$. The agent solves his problem by responding to this announced contract from the principal. Similar to $\mathrm{(O-A')}$ and using $V_a(t,c_t) = h_A(c_t)=-c_t$, we obtain the corresponding HJB equations as
\begin{equation*}
\begin{split}
\min_{E_t}\left[ \frac{\partial V_a}{\partial c_t}\left( rc_t + \frac{1}{2}K_t^{\mathsf{T}}R_t^{-1}K_t + \Gamma_t^{\mathsf{T}}(E_t-R_t^{-1}K_t) \right) + f_A(t,p_t,E_t) \right]
+\frac{\partial V_a}{\partial t} = rV_a,\\
V_a(T,c_T) = -c_T.
\end{split}
\end{equation*}
The optimal solution of the agent is achieved at $$E_t^o = {\arg\min}_{E_t} - \Gamma_t^{\mathsf{T}} E_t + \frac{1}{2}E_t^{\mathsf{T}}R_tE_t,$$ which yields $E_t^o = R_t^{-1}\Gamma_t$. Based on Lemma \ref{lem_LQ_team}, we choose $\Gamma_t = K_t$, and thus the agent implements the team optimal solution $E_t^b$. Further, \eqref{implementable_ct} degenerates to the one in \eqref{lem_LQ_team_ct}. \qed
\end{proof}

\textit{Remark:} In the LQ setting under full information and incomplete information, the optimal contract and the manager's behavior do not relate to the risk volatility $\Sigma_t(Y_t)$ of the network. The reason is that the cost function of the principal is linear in the systemic risk $Y_t$. Hence, the expectation of the risk volatility term is zero, and $\Sigma_t(Y_t)$ does not play a role in the optimal dynamic contract. This fact in turn corroborates the zero information rent in the LQ setting due to the removal of risk uncertainty.

A more general class of scenarios satisfying the certainty equivalence principle that leads to zero information rent is summarized as follows.

\begin{corollary}
When $f_P(t,\phi,p_t)$, $h_p(\phi)$ and $h_A(\phi)$ are linear in the argument $\phi$, then $I_R=0$, where the optimal contracts under the full information and incomplete information coincide.
\end{corollary}
\begin{proof}
The linearity of functions removes the effects of risk uncertainties on the performance of the principal and the agent which leads to a zero information rent. \qed
\end{proof}

\section{Case Studies}\label{examples}
We demonstrate, in this section, the optimal design principles of dynamic contracts for systemic cyber risk management of enterprise networks through examples. Specifically, we first utilize a case study with one node to show that the dynamic contracts can successfully mitigate the systemic risk in a long period of time. Then, we investigate an enterprise network with a set of interconnected nodes to reveal the network effects in systemic risk management through dynamic contracts and discover a distributed way of mitigating the systemic risks.

\subsection{One-Node System Case}
First, we consider a one-dimensional case in which the enterprise network contains only one node, i.e., $Y_t$ is a scalar. Therefore, the risk manager protects the system by directing the security resources to this node. Note that for the LQ setting, 
the coupled ODEs in Theorem \ref{optimal_zeta_LQ} admit the unique solutions:
\begin{align}
K_t = \frac{\rho}{A-r}\left( (A-r+1)e^{(A-r)(T-t)} -1\right),\ 
m_t = \frac{K_t^2}{2rR_t}\left( e^{-r(T-t)}-1\right).
\end{align}
Therefore, based on Lemma \ref{lem_2}, the optimal effort of the risk manager is 
\begin{equation}
E_t^* = R_t^{-1}\zeta_t^*=\frac{\rho}{R_t(A-r)}\left( (A-r+1)e^{(A-r)(T-t)} -1\right),
\end{equation}
and the optimal compensation becomes
\begin{equation}
dc_t = \left( rc_t- \frac{K_t^2}{2R_t} + AK_tY_t\right) dt -  K_tdY_t,\ c_0=-\underline{J}_A.
\end{equation}

If the risk manager accepts this optimal contract, then the principal's excepted minimum cost is equal to 
$
J_P^*=K_0^{\mathsf{T}}y_0+m_0 - \underline{J}_A.
$

To illustrate the optimal mechanism design, we choose specific values for the parameters in Section \ref{LQ_case_section}: $\rho = 5\ \mathrm{k}\$/\mathrm{unit}$, $r=0.3$, $R_t=1.5\ \mathrm{k}\$/\mathrm{unit}^2$, $T=1$ year, $y_0=5$ unit, and $\underline{J}_A=-10\ \mathrm{k}\$$. Figure \ref{example_1d} shows the results for varying values of the parameter $A$. Note that a single node system with a larger $A$ indicates that it is more vulnerable and harder to mitigate the cyber risk. From Fig. \ref{example_1d}, we find that with a larger $A$, the system requires more effort from the risk manager to bring the cyber risk down to a relatively low level. In all cases, the effort decreases as time increases, and finally converges to a positive constant $\frac{\rho}{R_t}$. This phenomenon indicates that when the system risk is high, the agent should spend more effort in risk management. When the risk is reduced to a relatively low level and the system becomes secure, then less effort is preferable as the risk will not grow. In addition, the corresponding terminal compensation $c_T$ increases with the amount of effort spent.

\begin{figure}[t]
  \centering
  \subfigure[Effort]{
    \includegraphics[width=0.48\columnwidth]{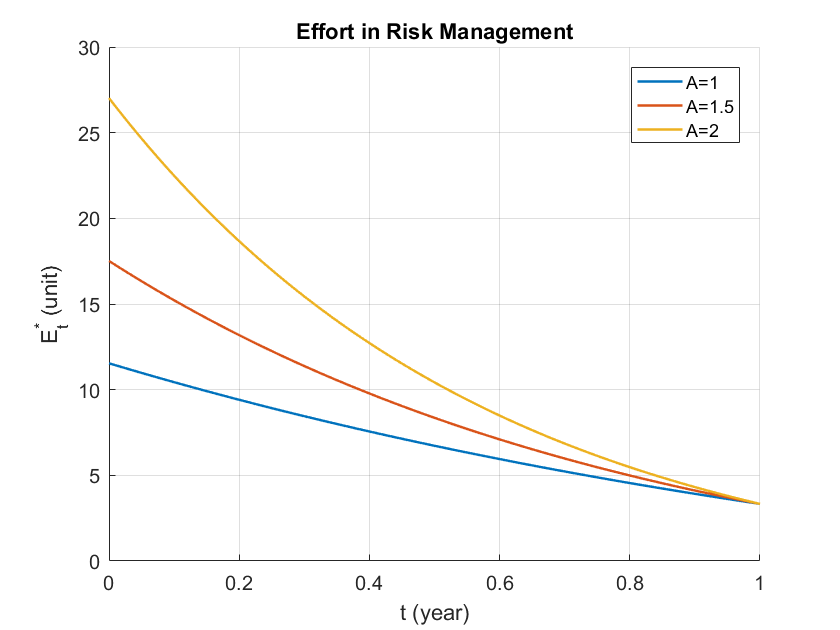}}
      \subfigure[Systemic cyber risk]{
    \includegraphics[width=0.48\columnwidth]{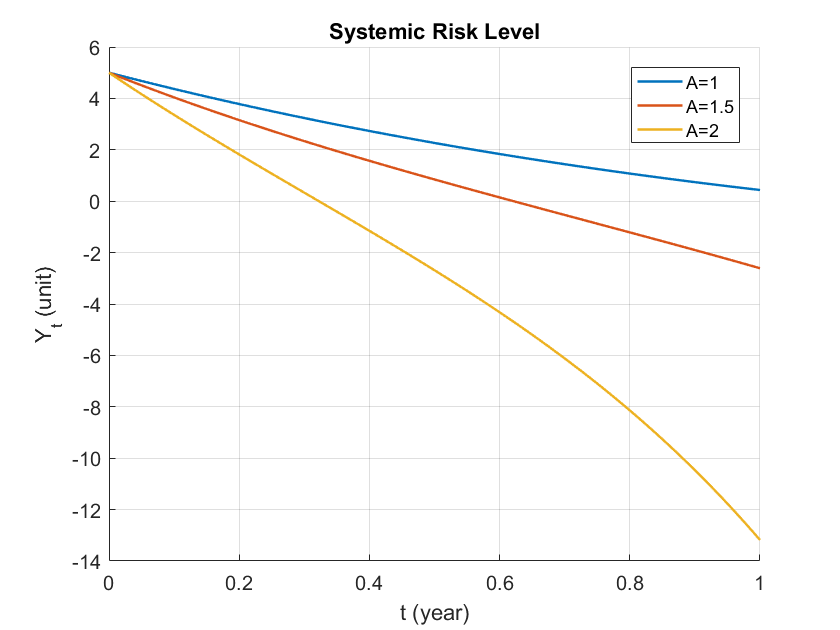}}
    	 \subfigure[Cumulative payment]{
    \includegraphics[width=0.5\columnwidth]{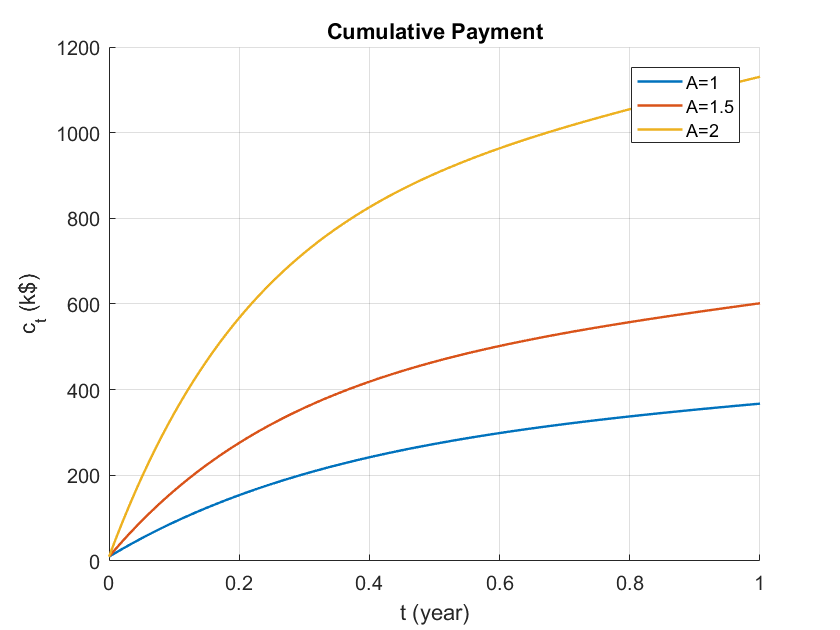}}
  \caption[]{(a), (b), and (c) show the effort,  the cyber risk and the terminal payment under the optimal contract. The terminal compensation $c_T$ increases with the spent effort of the risk manager.}
  \label{example_1d}
\end{figure}

\begin{figure}[t]
  \centering
  \subfigure[Effort]{
    \includegraphics[width=0.48\columnwidth]{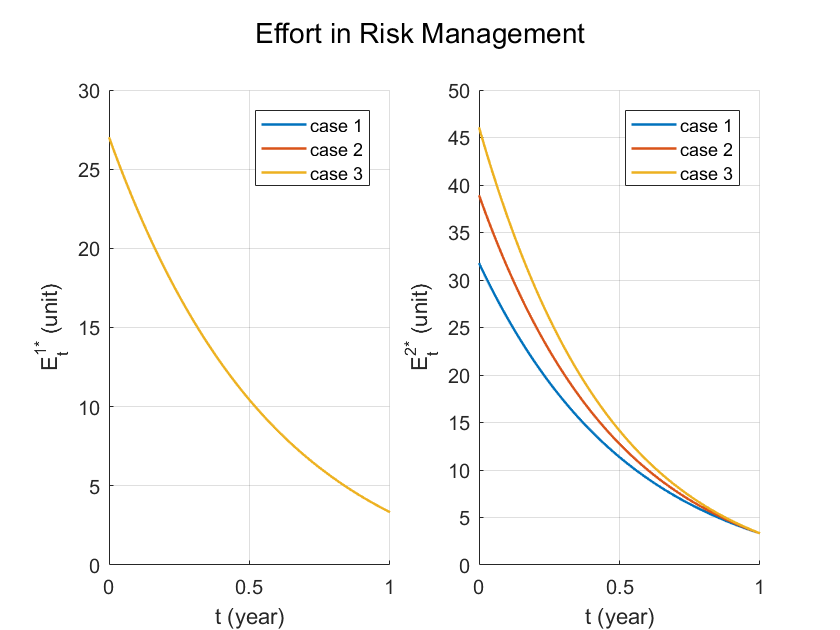}}
      \subfigure[Systemic cyber risk]{
    \includegraphics[width=0.48\columnwidth]{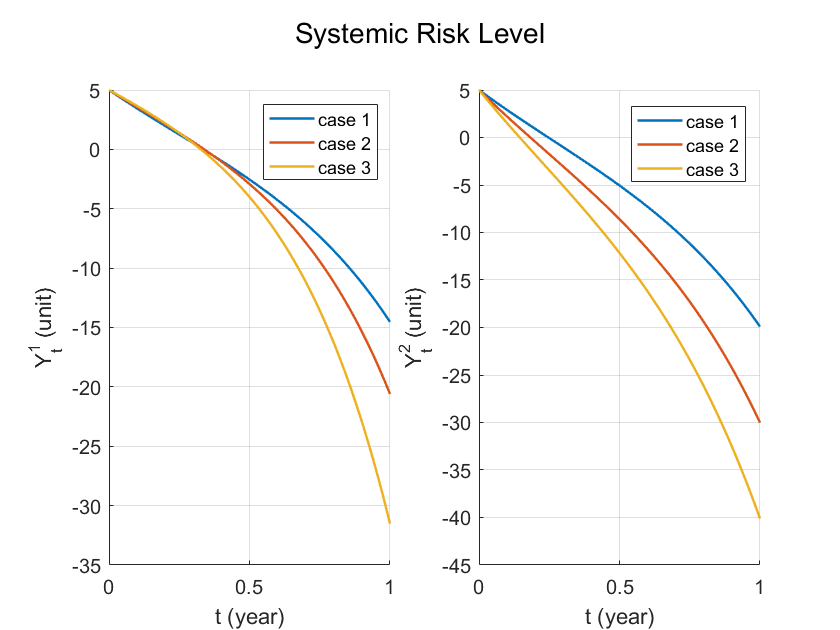}}
    	 \subfigure[Cumulative payment]{
    \includegraphics[width=0.5\columnwidth]{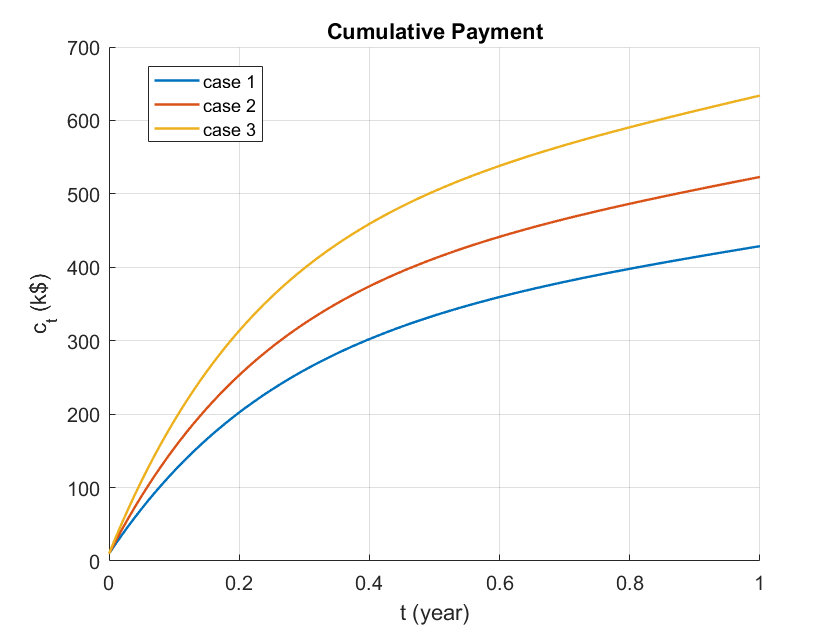}}
  \caption[]{(a), (b), and (c) show the effort, the systemic risk and the terminal payment under the optimal contract. Case 1: $A = [2,0.2;0,2]$; Case 2: $A = [2,0.5;0,2]$; Case 3: $A = [2,0.8;0,2]$. A higher network connectivity requires more effort to mitigate the systemic cyber risk.}
  \label{example_2d_1}
\end{figure}

\subsection{Network Case}
We next investigate cyber risk management over enterprise networks and characterize the interdependencies between nodes. The unique solutions to the ODEs in Theorem \ref{optimal_zeta_LQ} are then as follows:
\begin{align}
K_t &= \rho\left[ (A-rI)^{\mathsf{T}} \right]^{-1}\left(\left((A-rI)^{\mathsf{T}}+I\right)e^{(A-rI)^{\mathsf{T}}(T-t)}-I\right),\\
m_t &= \frac{K_t^{\mathsf{T}}R_t^{-1}K_t}{2r}\left( e^{-r(T-t)}-1 \right),
\end{align}
The optimal effort of the risk manager is 
\begin{align*}
E_t^* =R_t^{-1}\rho\left[ (A-rI)^{\mathsf{T}} \right]^{-1}\left(\left((A-rI)^{\mathsf{T}}+I\right)e^{(A-rI)^{\mathsf{T}}(T-t)}-I\right),
\end{align*}
and the optimal compensation follows \eqref{dc_lem_2}.

We first consider a cyber network containing two connected nodes. The system parameters are chosen as $\rho = [5;5] \ \mathrm{k}\$/\mathrm{unit}$, $r=0.3$, $R_t=[1.5,0;0,1.5]\ \mathrm{k}\$/\mathrm{unit}^2$, $T=1$ year, $y_0=[5;5]$ unit, and $\underline{J}_A=-10\ \mathrm{k}\$$. Moreover, we compare three scenarios in terms of network interdependencies. Specifically, we have case 1: $A = [2,0.2;0,2]$, case 2: $A = [2,0.5;0,2]$, and case 3: $A = [2,0.8;0,2]$. Figure \ref{example_2d_1} shows the results, where we denote by $E_t^{i*}$ and $Y_t^{i}$ the effort and the corresponding risk of node $i$, $i=1,2$, respectively. Similar to the single-node case, both the effort and systemic risk decrease over time. Specifically, the dynamic effort converges to $R_t^{-1}\rho$ which can be verified directly by the analytical expression. Comparing $E_t^{1*}$ with $E_t^{2*}$, we find that the risk manager should spend more effort on the nodes which can heavily influence other nodes. Even though there is no risk influence from node 1 to node 2, the optimal effort $E_t^{2*}$ increases as the influence strength becomes larger from node 2 to node 1. This phenomenon is consistent with the idea of \textit{controlling the origin} to constrain the propagation of cyber risks. Furthermore, the value of $E_t^{2*}$ indicates that a higher network connectivity requires more effort to mitigate the systemic cyber risk.

We next investigate a 4-node system where the network structures are shown in Fig. \ref{4_node}. The system parameters are the same as those in the 2-node case except for the matrix $A$. The diagonal entries in $A$ are all equal to 2 and the off-diagonal entries that correspond to a link are all equal to 0.2.  Figure \ref{example_4d} shows the results under the optimal mechanism. The risk manager spends more effort on node 1 in cases 2 and 3 than in case 1, as the risk of node 1 can propagate to node 4 in the former two cases. Another key observation is that the amount of allocated effort on each node mainly depends on its risk influences on other nodes rather than on the exogenous risks (node's outer degree), yielding a \textit{self-accountable} risk mitigation scheme. For example, even though node 4 impacts node 2 in case 3, the risk management efforts on node 2 are close in cases 2 and 3. A similar pattern can be seen on node 4 in cases 1 and 2. This observation provides a distributed method of risk management which reduces the complexity of decision-making by simplifying the network structures and classifying the nodes based on their outer degrees.  By comparing three cases, we also conclude that more complex cyber interdependencies induce higher cost on the principal in the security investment.

\begin{figure}[t]
  \centering
    \includegraphics[width=.7\columnwidth]{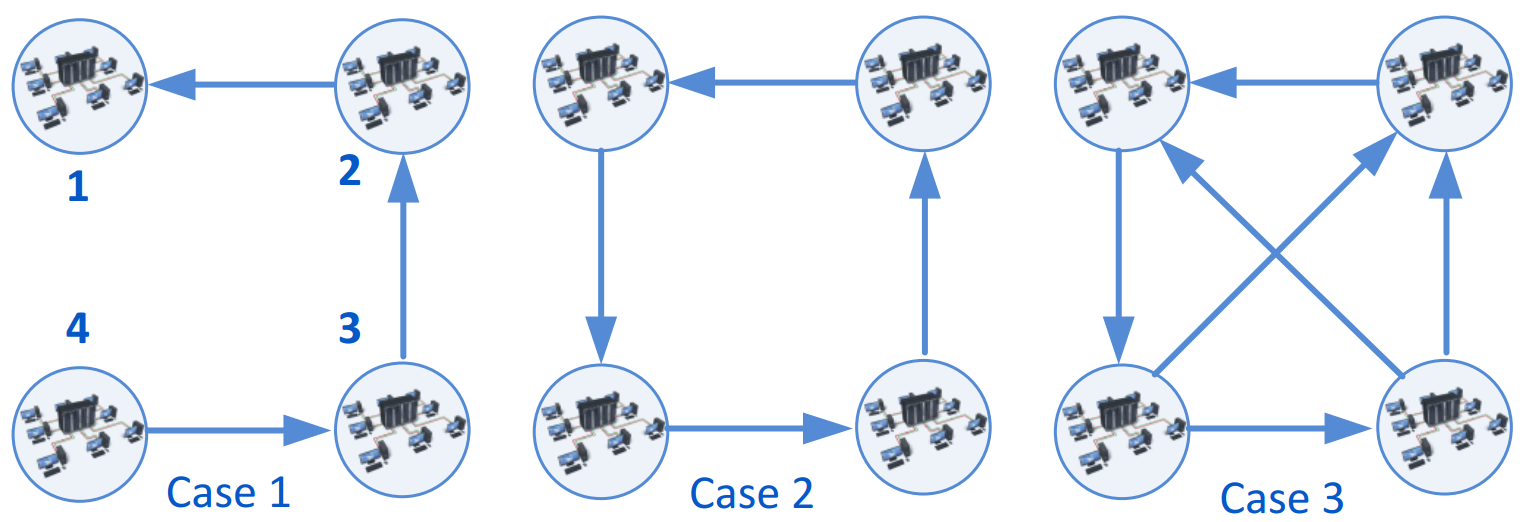}
  \caption[]{Three different structures of enterprise network. The risk influence strengths are the same, admitting a value of 0.2 in matrix $A$. }
  \label{4_node}
\end{figure}

\begin{figure}[t]
  \centering
  \subfigure[Effort]{
    \includegraphics[width=0.47\columnwidth]{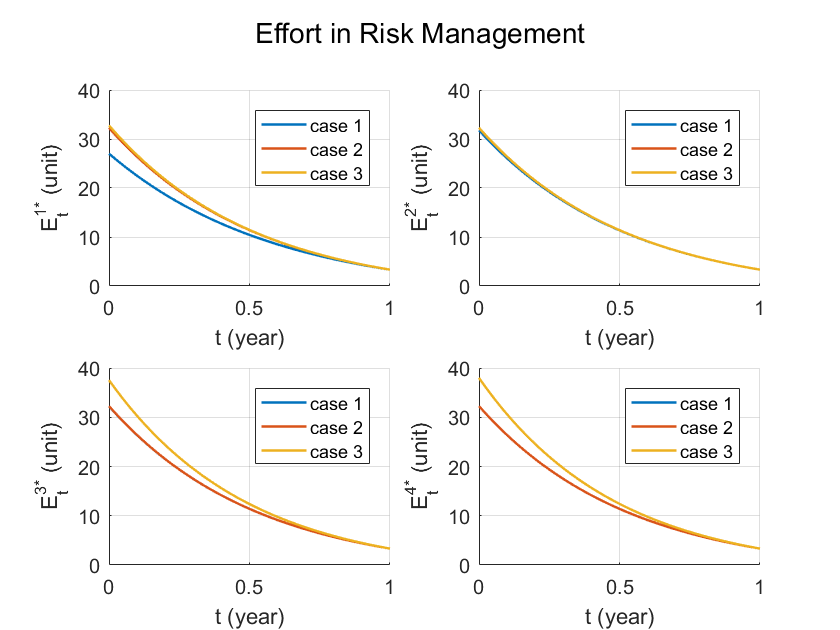}}
      \subfigure[Systemic cyber risk]{
    \includegraphics[width=0.47\columnwidth]{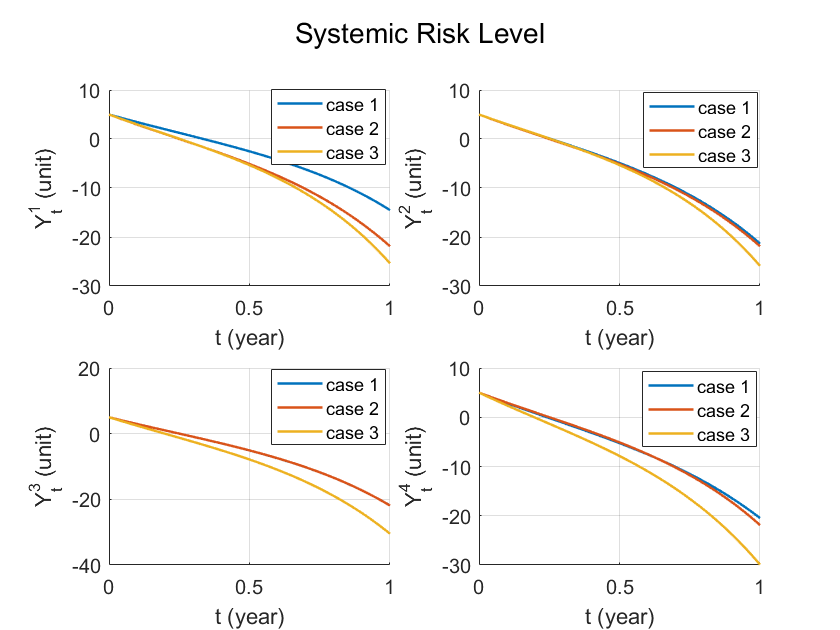}}
    	 \subfigure[Cumulative payment]{
    \includegraphics[width=0.5\columnwidth]{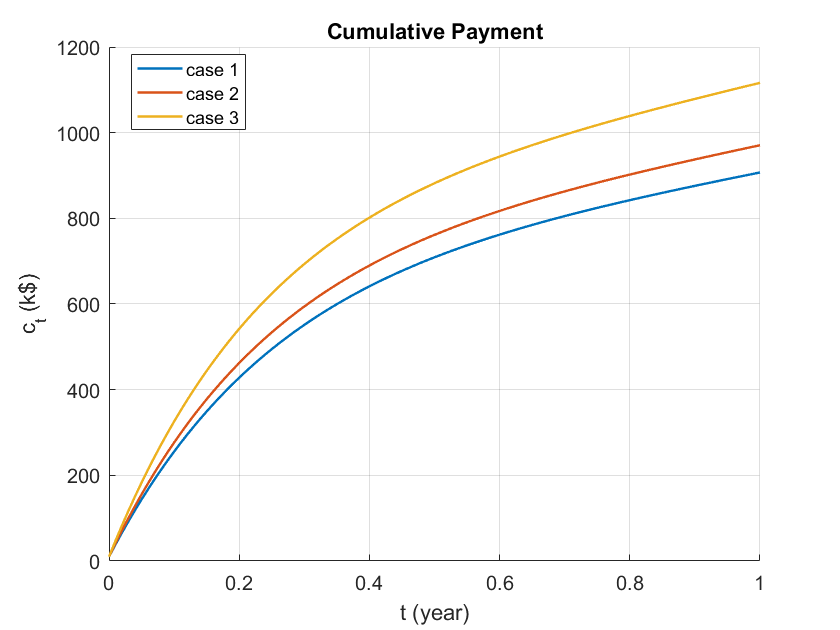}}
  \caption[]{(a), (b), and (c) show the effort, the systemic risk and the terminal payment under the optimal contract. Each node is self-accountable for its risk influence on others.}
  \label{example_4d}
\end{figure}

\begin{figure}[!ht]
  \centering
  \subfigure[]{
    \includegraphics[width=0.45\columnwidth]{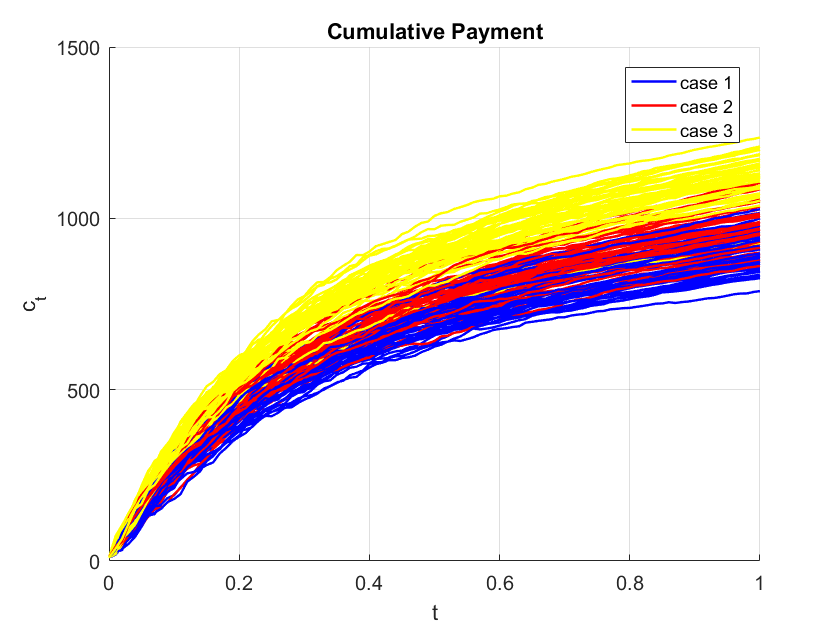}\label{payment_risk_1}}
      \subfigure[]{
    \includegraphics[width=0.45\columnwidth]{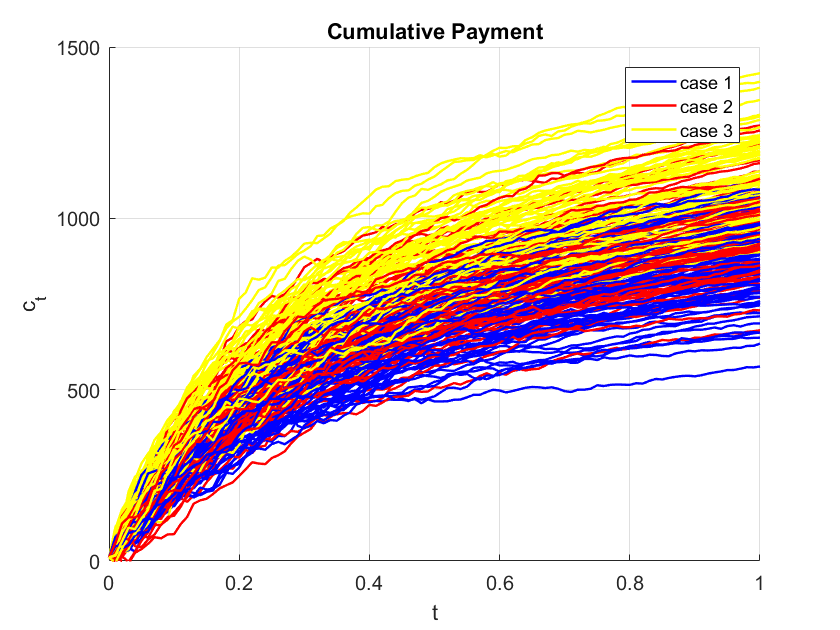}\label{payment_risk_2}}
  \caption[]{(a) and (b) depict the optimal terminal payment under different risk volatility structure. The risk volatility of nodes is independent in (a), while the influence  of risk volatility in (b) admits a cycle structure as case 2 in Fig. \ref{4_node}. The results indicate that a larger interdependency of cyber risk volatility yields compensation schemes with a larger variance. }
  \label{payment_risk}
\end{figure}

Note that in the above case studies, all variables were evaluated under the expectation with respect to the cyber risk uncertainty. As shown in Corollary \ref{volatility_impact}, even though the expected compensation is independent of the network risk uncertainty, the actual compensation during contract implementation is influenced by the volatility term $\Sigma_t(Y_t)$. We present two scenarios in Fig. \ref{payment_risk}, where Fig. \ref{payment_risk_1} and Fig. \ref{payment_risk_2} are the compensation realizations under $\Sigma_t(Y_t) = I$ and $\Sigma_t(Y_t) = [1,1,0,0;0,1,1,0;0,0,1,1;1,0,0,1]$, respectively. When the nodes' risks face more sources of uncertainties in Fig. \ref{payment_risk_2}, the corresponding payment exhibits a larger variance comparing with the one in Fig. \ref{payment_risk_1}, which is consistent with the result of Corollary \ref{volatility_impact}.

\section{Conclusion}\label{conclusion}
In this paper, we have addressed the problem of  dynamic systemic cyber risk management of enterprise networks, where the principal provides contractual incentives to the manager, which include the compensations of direct cost of effort and indirect cost from risk uncertainties. This has involved a stochastic Stackelberg differential game with asymmetric information in a principal-agent setting.  Under the optimal incentive compatible scheme we have designed, the principal has rational controllability of the systemic risk where the suggested and adopted efforts coincide, and the risk manager's behavior is strategically neutral, depending only on the current net cost. Under mild conditions, we have obtained a separation principle where the effort estimation and the remuneration design can be separately achieved. We further have revealed a certainty equivalence principle for a class of dynamic mechanism design problems where the information rent is equal to zero. Through case studies, we have identified the network effects in the systemic risk management where the connectivity and node's outer degree play an important role in the decision making. Future work on this topic would consider cyber risk management of enterprise networks under Markov jump risk dynamics.

%\begin{acknowledgements}
%\end{acknowledgements}


\bibliographystyle{spmpsci}      % mathematics and physical sciences
\bibliography{references}   % name your BibTeX data base

%\bibliographystyle{IEEEtran}
%\bibliography{IEEEabrv,references}

\end{document}